\documentclass[aps,onecolumn,pra,superscriptaddress,notitlepage, longbibliography]{revtex4-1}

\usepackage[dvipdfmx]{graphicx}
\usepackage{amsmath,amssymb,dsfont}
\usepackage{braket}
\usepackage{bm}
\usepackage{enumerate}
\usepackage{amsthm}

\usepackage{enumerate}
\usepackage{physics}
\usepackage{mathtools}

\usepackage{bbm}
\usepackage{enumitem}
\usepackage{multirow}
\usepackage{url}
\usepackage{xparse}
\usepackage[colorlinks,linkcolor=blue,citecolor=blue]{hyperref}

\newcommand{\black}[1]{\textcolor{black}{#1}}

\newcommand{\R}{\mathbb{R}}
\newcommand{\C}{\mathbb{C}}
\newcommand{\floor}[1]{\left\lfloor #1 \right\rfloor}
\newcommand{\LaurentR}[1]{\R \qty[{#1}, {#1}^{-1}]}
\newcommand{\LaurentC}[1]{\C \qty[{#1}, {#1}^{-1}]}
\DeclareMathOperator{\Parity}{Parity}

\newtheorem{theorem}{Theorem}

\begin{document}
\title{Robust Angle Finding for Generalized Quantum Signal Processing}

\author{Shuntaro Yamamoto}\email{shun0923@g.ecc.u-tokyo.ac.jp}
\affiliation{Department of Applied Physics, University of Tokyo, 7-3-1 Hongo, Bunkyo-ku, Tokyo 113-8656, Japan}

\author{Nobuyuki Yoshioka}\email{nyoshioka@ap.t.u-tokyo.ac.jp}
\affiliation{Department of Applied Physics, University of Tokyo, 7-3-1 Hongo, Bunkyo-ku, Tokyo 113-8656, Japan}

\begin{abstract}
Quantum Signal Processing (QSP), together with the quantum singular value transformation, is one of the central quantum algorithms due to its efficiency and generality in many fields including quantum simulation, quantum machine learning, and quantum cryptography.
The largest bottleneck of QSP and its family is its difficulty in finding the phase angle sequence for signal processing.
We find that this is in particular prominent when one employs the generalized formalism of the QSP, or the GQSP, to employ arbitrary single-qubit unitaries for signal processing operator. In this work, we extend the framework of GQSP and propose a robust angle finding algorithm.
The proposed angle finding algorithm, based on Prony's method, successfully generates angle sequence of precision $10^{-13}$ up to polynomial degrees of hundreds within a second.
By applying our method to Hamiltonian simulation, we find that the number of calls, or queries, to signal operators are essentially halved compared to the ordinary framework of QSP.
\end{abstract}

\maketitle

\section{Introduction}\label{sec:intro}
Quantum Signal Processing (QSP) and Quantum Singular Value Transformation (QSVT) are among the most general quantum algorithms that unifies numerous existing algorithms~\cite{low2016methodology, low2017optimal, Low2019hamiltonian, Gily_n_2019, Martyn_2021}.
In QSP, one first block-encodes non-unitary matrices using  ancilla qubits, and then perform a polynomial transformation on each qubitized subspace to realize the desired polynomial of matrix. 
It has been pointed out that application of QSP yields query-optimal Hamiltonian simulation algorithms~\cite{low2017optimal, Mizuta2023optimalhamiltonian}, and has also been considered in the context of practical quantum resource estimation~\cite{childs2018toward, yoshioka2022hunting, beverland2022assessing, dalzell2023quantum, sakamoto2023end}.

Although the QSP/QSVT has been continuously developed, improved, and generalized, these algorithms bear the notorious difficulty in determining the concrete circuit structure. Specifically,
 the block-encoded signal operator must be interleaved by signal processing operators with carefully tuned angle sequence, which cannot be determined analytically except for very limited target functions~\cite{mizuta2024recursive}. 
It was initially considered that the difficulty---or classical computation error---in angle finding procedure might practically negate the advantages of QSP/QSVT algorithms~\cite{childs2018toward}. However, recent works have successfully determined the angle sequence for polynomials of orders with hundreds to thousands~\cite{Haah_2019, Dong_2021}.
While the runtime of the algorithm by Haah~\cite{Haah_2019} scales $O(d^3)$ with $d$ being the degree of polynomial, the space complexity of $\tilde{O}(d)$ is considered to be unstable compared to those that require $O({\rm polylog}(d))$, which are viewed as stable~\cite{higham2002accuracy}.
Ref.~\cite{Dong_2021} focuses on the ordinary QSP that employs Pauli rotation as the signal processing operators. This algorithm heavily exploits the symmetry in angle sequence, a feature  not present in the generalized framework of QSP (GQSP)~\cite{motlagh2023generalized} in which the representability of the QSP is significantly extended by employing arbitrary SU(2) single qubit unitary instead of Pauli rotations.

In this work, we make two contributions: (i) extension of the GQSP framework, and (ii) the proposal of a robust angle-finding algorithm that successfully generates the angle sequence for approximation up to extremely high accuracy that is limited by the machine precision.
This angle-finding algorithm represents a strict generalization of the approach presented in Ref.~\cite{Dong_2021}, in that it also relies on the Prony's method for its operation.
Through numerical benchmarking for the approximation of $f(x)=e^{-i\tau x}$--- simulating Hamiltonian dynamics over time $\tau$---we demonstrate that our proposed algorithm can reliably generate angle sequences for GQSP of length in hundreds, up to accuracy of  $10^{-13}$, which is limited by the machine precision during the algorithm, within a second. Furthermore, this generalization enables us to halve the implementation cost of Hamiltonian simulation compared to the ordinary QSP,  in terms of the number of calls to the controlled signal operators.

\section{Preliminaries}
The QSP is a procedure that performs a polynomial transformation $f(W)$ of degree $d$ on a block-encoded operator $W$, interleaving it with signal processing operators~\cite{low2016methodology, low2017optimal, Low2019hamiltonian, Gily_n_2019}.
Specifically, for a signal operator parameterized as $W(\theta)$ with $\theta\in \mathbb{R}$, the objective is to perform $\theta \mapsto f(\theta)$.
Here, following the structure of Appendix A in ~\cite{Martyn_2021}, we introduce two variants of the QSP: $(W_x, S_z)$-QSP and $(W_z, S_x)$-QSP. 
Both variants require identifying suitable angle parameters for the signal process operators. The conditions for determining these parameters are rigorously stated in Theorem~\ref{thm:Wx_QSP_PQ_cond},~\ref{thm:Wx_QSP_f_cond} for the $(W_x, S_z)$-QSP and Theorems~\ref{thm:Wz_QSP_FG_cond},~\ref{thm:Wz_QSP_f_cond} for the $(W_z, S_x)$-QSP.

\subsection{$(W_x, S_z)$-QSP}
The $(W_x, S_z)$-QSP employs $x$-rotation as the signal operator and $z$-rotation as the signal processing operator.
Concretely, we assume that the block encoding is done in a way such that the Hilbert space is {\it qubitized} into direct sum of two-dimensional subspace, and assume that the operations on one of the subspaces are described as
\begin{align}
W_x(\theta)
&\coloneqq 
e^{i \theta X}
= \mqty(x & i \sqrt{1 - x^2} \\ i \sqrt{1 - x^2} & x), \\
S_z(\phi)
&\coloneqq e^{i \phi Z}
= \mqty(e^{i \phi} & 0 \\ 0 & e^{- i \phi}),
\end{align}
where $x\coloneqq\cos\theta \in [-1, 1]$ for $\theta \in \mathbb{R}$ and $X, Z$ are $x, z$ components of Pauli matrices.
Correspondingly, the QSP operation sequence $U_{\vb*{\Phi}}$ for $\vb*{\Phi}\coloneqq (\phi_0, ..., \phi_d)\in\mathbb{R}^{d+1}$ is defined as
\begin{equation}
U_{\vb*{\Phi}}(x)
\coloneqq S_z(\phi_0) \prod_{k = 1}^d W_x(\theta) S_z(\phi_k).
\end{equation}

The goal is to construct the angle sequence $\vb*{\Phi}$ such that we can approximate the transformation as $\ev{U_{\vb*{\Phi}}(x)}{0} = f(\theta)$.
Low and Chuang showed the existence condition of such an angle sequence~\cite{Low2019hamiltonian}, and later constructive method was provided by Gily\'{e}n {\it et al.} in Ref.~\cite{Gily_n_2019}. 
Assuming that the target unitary can be written using complex polynomials $P, Q\in C[x]$ as follows,
\begin{equation}\label{eq:wx_sequence}
U_{P Q}(x)
= \mqty(P(x) & i Q(x) \sqrt{1 - x^2} \\ i Q^*(x) \sqrt{1 - x^2} & P^*(x)),
\end{equation}
then, the following theorems hold:

\begin{theorem} 
\label{thm:Wx_QSP_PQ_cond}(Theorem 3 in~\cite{Gily_n_2019})
$\forall d\in \mathbb{N}, \exists \Phi = (\phi_0, \dots, \phi_d) \in \mathbb{R}^{d+1}$ s.t. \black{$U_{\vb*{\Phi}}(x) = U_{P Q}(x)$ for $x\in[-1, 1]$}
if and only if
\begin{enumerate}[label=(\roman*)]
\item
$\deg(P) \leq d, \deg(Q) \leq d - 1$.
\item
${\rm Parity}(P)=d \bmod 2$, ${\rm Parity}(Q)=d-1 \bmod 2$.
\item
$\forall x \in [-1, 1], \abs{P(x)}^2 + (1 - x^2) \abs{Q(x)}^2 = 1$.
\end{enumerate}
\end{theorem}

In practice, one may desire to perform real polynomial transformation $f(x) \coloneqq f(\theta=\arccos x)$. In such a case, it suffices to find $P, Q$ such that ${\rm Re}[P] =f, {\rm Re}[Q]=0$ since $\ev{U_{P Q}}{+} = f(x)$. The existence condition for phase sequence can be stated rigorously as follows:

\begin{theorem}
\label{thm:Wx_QSP_f_cond}~(Theorem 5 in~\cite{Gily_n_2019})~
$\forall d\in \mathbb{N}, \exists \Phi \in \mathbb{R}^{d+1}$ s.t. $U_{\vb*{\Phi}}(x) =U_{PQ}(x)$ for \black{$x\in [-1, 1]$}, ${\rm Re}[P]=f, {\rm Re}[Q]=0$ if and only if
\begin{enumerate}[label=(\roman*)]
\item
$\deg(f) \leq d$.
\item
${\rm Parity}(f) = d\bmod 2$.
\item
$\forall x \in [-1, 1], \abs{f(x)} \leq 1$.
\end{enumerate}
\end{theorem}
\noindent
\black{Note that this Theorem holds even for ${\rm Re}[Q] \neq 0$ when $({\rm Re}[P])^2 + \qty(1-x^2)({\rm Re}[Q])^2 \leq 1$ is satisfied~\cite{Gily_n_2019}.}

\subsection{$(W_z, S_x)$-QSP} \label{subsec:Wz_convention}
Next, we introduce the $(W_z, S_x)$-QSP which in turn employs  $z$-rotation as the signal operator and $x$-rotation as the signal processing operator.
Let $U(1) \coloneqq \qty{w \in \C \colon \abs{w} = 1}$ denote the unit circle in the complex plane, and let $w \coloneqq e^{i \theta} \in U(1)$ be defined for $\theta \in \R$. Then, the two operators are defined as follows,
\begin{align}
W_z(\theta)
&\coloneqq e^{i \theta Z}
= \mqty(w & 0 \\ 0 & w^{-1}), \\
S_x(\phi)
&\coloneqq e^{i \phi X}
= \mqty(\cos \phi & i \sin \phi \\ i \sin \phi & \cos \phi).
\end{align}
In similar to $(W_x, S_z)$-QSP, the QSP operation sequence $U_{\vb*{\Phi}}$ for $\vb*{\Phi}=(\phi_0, ..., \phi_d)\in \mathbb{R}^{d+1}$ is defined as
\begin{equation}
U_{\vb*{\Phi}}(w)
\coloneqq S_x(\phi_0) \prod_{k = 1}^d W_z(\theta) S_x(\phi_k),
\end{equation}
which \black{is related with Eq.~\eqref{eq:wx_sequence}} as $U_{\vb*{\Phi}}(x) = H U_{\vb*{\Phi}}(w) H$.

Now the problem is to find the appropriate angle sequence $\vb*{\Phi}$ such that $\ev{U_{\vb*{\Phi}}(w)}{0} = f(\theta)$.
 Let us consider a unitary that can be written using Laurent polynomial with real coefficients $F, G \in \LaurentR{w}$ as 
\begin{equation}
U_{F G}(w)
= \mqty(F(w) & i G(w) \\ i G \qty(w^{-1}) & F \qty(w^{-1})).
\end{equation}
Then, one can prove that following Theorems~\ref{thm:Wz_QSP_FG_cond} and~\ref{thm:Wz_QSP_f_cond} holds:

\begin{theorem}(\cite{Haah_2019, chao2020finding}, Appendix A in~\cite{Martyn_2021})
\label{thm:Wz_QSP_FG_cond}
$\forall d\in \mathbb{N}, \exists \vb*{\Phi}\in\mathbb{R}^{d+1}$ s.t. \black{$U_{\vb*{\Phi}}(w) = U_{FG}(w)$ for $w \in U(1)$} if and only if
\begin{enumerate}[label=(\roman*)]
\item
$\deg(F) \leq d, \deg(G) \leq d$.
${\rm Parity}(F) = {\rm Parity}(G) = d \bmod 2$.
\item
$\forall w \in U(1), \abs{F(w)}^2 + \abs{G(w)}^2 = 1$.
\end{enumerate}
\end{theorem}

\begin{theorem}(\cite{Haah_2019, chao2020finding}, Appendix A in~\cite{Martyn_2021})
\label{thm:Wz_QSP_f_cond}
Let $f(w)\coloneqq f(\theta = \arg w)$ be a Laurent polynomial with real coefficients. Then, $\forall d\in \mathbb{N}, \exists \vb*{\Phi}\in \mathbb{R}^{d+1}$ s.t. $F = f$, \black{$U_{\vb*{\Phi}}(w) = U_{FG}(w)$ for $w \in U(1)$} if and only if
\begin{enumerate}[label=(\roman*)]
\item
$\deg(f) \leq d$.
\item
${\rm Parity}(f)= d \bmod 2$.
\item
$\forall w \in U(1), \abs{f(w)} \leq 1$.
\end{enumerate}
\end{theorem}

\section{Angle finding for Quantum Signal Processing}\label{sec:angle_finding}

\subsection{Overview} \label{subsec:overview_angle_finding}
First we provide an overview for constructing signal operators for a desired polynomial transformation (see Fig.~\ref{fig:flow_chart} for graphical visualization). 
The main steps are three-fold: choosing the QSP basis, truncating/partitioning the target function, and computing phase factors for truncated functions.
In the second step,  operations known as {\it truncation} and {\it partition} are repeated to generate a sequence of functions $\{f_i\}$, each of which undergoes the third step to approximate $f$ as a linear combination over $f_i$.
As we provide details in Sec.~\ref{subsec:computation_angles}, the third step can be executed mainly in two ways: the algebraic approach (direct method) and the numerical approach (optimization method).

\begin{figure}[t]
    \begin{center}
        \includegraphics[width=0.95\linewidth]{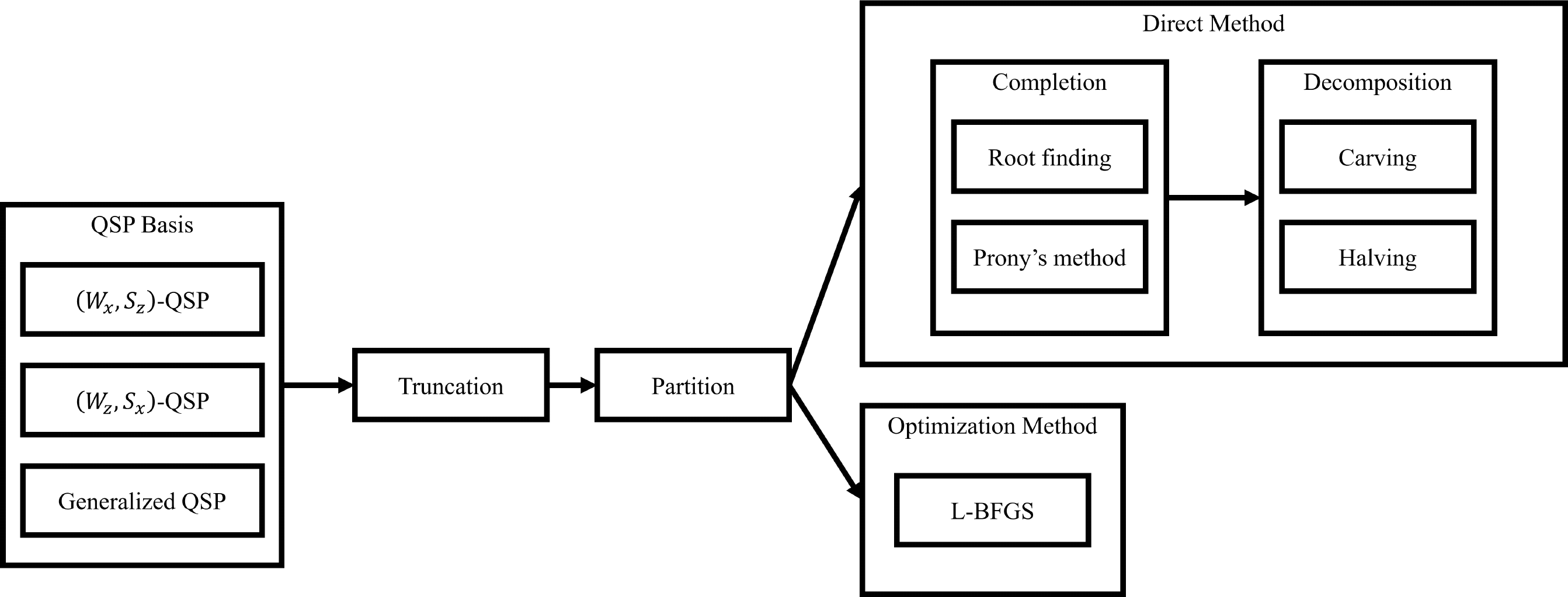}
        \caption{Flow chart of phase angle finding for QSP/GQSP algorithms.}
        \label{fig:flow_chart}
    \end{center}
\end{figure}

\subsection{Truncation and Partition}\label{subsec:truncation}
The objective of the second step is to decompose the target function $f$ into a linear combination of $\{f_i\}$.
First, one performs an operation called the {\it truncation}, which approximates the target function $f$ with a complex-coefficient polynomial $\tilde{f}$ of order not exceeding $d$.
This approximation is typically achieved using the Remez exchange algorithm~\cite{remez1934}, while some functions permit analytic truncation, such as  Hamiltonian simulation discussed in Sec.~\ref{sec:numerics}. Practically, one shall be aware that this step introduces an error into the approximation. For further details, refer to Sec.~\ref{sec:numerics} for details which indicates that the truncation error dominates the entire error at low $d$ regime.

The truncation is followed by an operation called the {\it partition}, which expresses $\tilde{f}$ as a linear combination of functions that satisfy the conditions provided in Theorem~\ref{thm:Wx_QSP_f_cond} or~\ref{thm:Wz_QSP_f_cond} (see Sec.~\ref{sec:angle_finding_gqsp} for the case of GQSP).
For instance, in the case of $(W_x, S_z)$-QSP, we decompose as
\begin{equation}
\tilde{f}(x) = \alpha \qty(f_1(x) + f_2(x) + i f_3(x) + i f_4(x)),
\end{equation}
where $f_1, f_3 \in \R[x]$ are even functions and $f_2, f_4 \in \R[x]$ are odd functions. Note that $\alpha >0$ is taken sufficiently large so that $\forall x \in [-1, 1], \abs{f_i(x)}<1~(i \in \{1,2,3,4\})$ is satisfied.
The partition for the $(W_z, S_x)$-QSP is done using $f_1, f_2, f_3, f_4 \in \LaurentR{w}$ as 
\begin{equation}
\tilde{f}(w) = \alpha \qty(f_1(w) + f_2(w) + i f_3(w) + i f_4(w)),
\end{equation}
where $f_1, f_3$ are even functions, $f_2, f_4$ are odd functions, and $\alpha > 0$ is a normalization factor such that $\forall w \in U(1), \abs{f_i(w)}<1~    (i \in \{1,2,3,4\})$.

\subsection{Computation of phase angles} \label{subsec:computation_angles}
After the operations of truncation and partition, the problem now reduces to finding the phase factors that satisfy Theorem~\ref{thm:Wx_QSP_f_cond} or~\ref{thm:Wz_QSP_f_cond}.
\black{Now} we proceed to review explicit algorithms that compute the angle sequence.
\black{While we describe the methods only briefly here for the sake of conciseness, we provide the details in Appendix~\ref{app:angle_finding_Wx} and~\ref{app:angle_finding_Wz} for $(W_x, S_z)$ and $(W_z, S_x)$-QSP, respectively.}

\subsubsection{Direct method} \label{subsubsec:direct}
Direct methods for angle computation commonly consist of two substeps called  {\it completion} and {\it decomposition}.
In the completion procedure for $(W_x, S_z)$-QSP~(or $(W_z, S_x)$-QSP), one computes the polynomials that satisfy the conditions of Theorem~\ref{thm:Wx_QSP_PQ_cond}~(Theorem~\ref{thm:Wz_QSP_FG_cond}), given a function $f$ that satisfies the conditions stated in Theorem~\ref{thm:Wx_QSP_f_cond}~(Theorem~\ref{thm:Wz_QSP_f_cond}).
One of the most common approaches considered in the literature employs the {\it root finding}, in which one numerically computes all the roots of $1 - \abs{f}^2$, relying on the fact that polynomials $P, Q$ (or $F, G$) can be expressed explicitly using such roots~(\cite[Section 3.1]{Gily_n_2019}, \cite[Section 8]{chao2020finding}).
It has been pointed out in Ref.~\cite{chao2020finding} that, by randomly choosing the root used for computation of $G$, one may enhance the accuracy of the solution.
In the case of $(W_z, S_x)$-QSP, one may alternatively use the Prony's method, which performs the completion by directly computing the Laurent polynomial with real coefficients $G$ \black{such that} $\abs{G}^2=1 - \abs{f}^2$, without relying on the roots~\cite{Ying_2022}.

In the decomposition step, one computes the phase factors $\vb*{\Phi}$ from the polynomials $P, Q$ that satisfies conditions in Theorem~\ref{thm:Wx_QSP_PQ_cond} (or $F, G$  satisfying Theorem~\ref{thm:Wz_QSP_FG_cond}).
Two representative algorithms are the carving~\cite{Gily_n_2019} and  halving~\cite{chao2020finding}.
The carving algorithm first determines \black{the last angle} $\phi_d$ by examining the highest degree coefficient of a polynomial of degree $d$ or less, and then reduces the problem to one involving a pair of polynomials of degree $d - 1$ or less to solve it (\cite[Section 3.1]{Gily_n_2019}). On the other hand, the halving algorithm iteratively decomposes a QSP operation sequence of degree $d$ into the product of two QSP operation sequences of degree $d / 2$ by solving a system of linear equations, eventually breaking it down into the product of $d$ first-degree QSP operation sequences (\cite[Section 8]{chao2020finding}). In the case of the halving algorithm, we can enhance the numerical stability in solving the system of linear equations by adding a constant $\epsilon_{\textrm{cap}}$ of the order of the allowable error to the highest and lowest degree coefficients of the polynomial $f(w)$. Such an operation is known as {\it capitalization} (\cite[Section 8]{chao2020finding}).

\subsubsection{Optimization method} 
\label{subsubsec:optimization}

\black{One may also compute the phase angle sequence $\vb*{\Phi}$ via numerical optimization~\cite{Dong_2021, dong2022infinite, dong2024robust}.
In Ref.~\cite{Dong_2021}, it was proposed that one may use the L-BFGS method, which is one of the most widely-used quasi-Newton methods in operations research, in order to optimize the following loss function for a given function $f$ that satisfies the conditions in Theorem~\ref{thm:Wx_QSP_f_cond}:
\begin{equation}
L(\vb*{\Phi}) = \frac{1}{\tilde{d}} \sum_{j = 1}^{\tilde{d}} \abs{\Re[\ev{U_{\vb*{\Phi}}(x_j)}{0}] - f(x_j)}^2,
\end{equation}
where $\tilde{d} = \left\lceil \frac{d + 1}{2} \right\rceil, x_j = \cos(\frac{(2 j - 1) \pi}{4 \tilde{d}})$.
It is noteworthy that the symmetry structure can be exploited; it can be proved that there exists a global minima $L(\vb*{\Phi})=0$ with $\vb*{\Phi}$ with inversion symmetry of $(\phi_0, ..., \phi_d) = (\phi_d, ..., \phi_0).$
As we later mention in Sec.~\ref{sec:numerics}, such a symmetry is absent in GQSP, in which the performance is heavily degraded compared to the case of ordinary QSP. The sole difference is that the summation over $j$ is taken for $1 \leq j\leq d$ instead of $1 \leq j \leq \tilde{d}.$
}

\section{Generalized Quantum Signal Processing} \label{sec:gqsp}
While the QSP has already successfully unified quantum algorithms, it has been pointed out in Ref.~\cite{motlagh2023generalized} that one may further extend its applicability by employing the SU(2) rotation instead of using solely the $x$- or $z$-rotations.
In this work, we propose a modification to the framework proposed \black{in the first draft of} Ref.~\cite{motlagh2023generalized} and find that we may halve the number of queries to the signal operators to realize the Hamiltonian simulation. 
Concretely, we use two signal operators $W_0$ and $W_1$ for signal processing operator $R$ given as SU(2) rotation, so that the QSP can be naturally applied to Laurent polynomials with complex coefficients. Let us define
\begin{align}
W_0(w)
&\coloneqq \mqty(w & 0 \\ 0 & 1) , \ \ \ \ 
W_1(w)
\coloneqq \mqty(1 & 0 \\ 0 & w^{-1}), \\
R(\theta_k, \phi_k, \lambda_k)
&\coloneqq \mqty(e^{i (\lambda_k + \phi_k)} \cos \theta_k & e^{i \lambda_k} \sin \theta_k \\ e^{i \phi_k} \sin \theta_k & - \cos \theta_k),
\end{align}
which yields the QSP operation sequence $U_{\vb*{\Theta} \vb*{\Phi} \lambda}$ as 
\begin{equation} \label{eq:gqsp_sequence}
U_{\vb*{\Theta} \vb*{\Phi} \lambda}(w)
\coloneqq
\begin{multlined}[t]
R(\theta_{- d_-}, \phi_{- d_-}, \lambda) 
\cdot \qty(\prod_{k = - d_- + 1}^0 W_1(w) R(\theta_k, \phi_k, 0)) 
\cdot \qty(\prod_{k = 1}^{d_+} W_0(w) R(\theta_k, \phi_k, 0)).
\end{multlined}
\end{equation}
In similar to the ordinary QSP, the problem  boils down to finding the appropriate phase factors $\vb*{\Theta}, \vb*{\Phi}, \lambda$ such that $\ev{U_{\vb*{\Theta} \vb*{\Phi} \lambda}(w)}{0} = f(\theta)$ is satisfied.
Our main findings is that existence of such phase factors can be assured when the target unitary can be written using complex-coefficient Laurent polynomials $F, G\in \LaurentC{w}$ as follows:
\begin{equation}
V_{F G}(w)
\coloneqq \mqty(F(w) & i G(w) \\ \cdot & \cdot).
\end{equation}
Let us denote the highest and lowest power with nonzero coefficients for a polynomial $P$ as $d_{{\rm max/min}, P}$.
The existence of angle sequences are stated rigorously by the following theorems.

\begin{theorem}
\label{thm:GQSP_FG_cond}
$\forall d_+, d_-\in \mathbb{N}, \exists \vb*{\Theta} = (\theta_{- d_-}, \ldots, \theta_{d_+}) \in \R^{d_+ + d_- + 1}, \vb*{\Phi} = (\phi_{- d_-}, \ldots, \phi_{d +}) \in \R^{d_+ + d_- + 1}, \lambda \in \R$ s.t. \black{$U_{\vb*{\Theta} \vb*{\Phi} \lambda}(w) = V_{F G}(w)$ for $w \in U(1)$} if and only if
\begin{enumerate}[label=(\roman*)]
\item
$d_{{\rm max},F}, d_{{\rm max},G} \leq d_+$ and $d_{{\rm min},F}, d_{{\rm min},G} \geq - d_-$.
\item
$\forall w \in U(1), \abs{F(w)}^2 + \abs{G(w)}^2 = 1$.
\end{enumerate}
\end{theorem}
\begin{proof}
The ``if" part of the theorem is proved by explicitly constructing an algorithm that is introduced in Sec.~\ref{sec:angle_finding_gqsp}.
The ``only if" part for the condition (i) follows from the fact that, by definition of $U_{\vb*{\Theta}, \vb*{\Phi},\lambda}$ that involves $d_+$ applications of $W_0$ and $d_-$ applications of $W_1$, the power for $w$ is restricted to $w^{-d_-}, ..., w^{d_+}.$ The ``only if" part for the condition (ii) follows from the unitarity of $U_{\vb*{\Theta} \vb*{\Phi} \lambda}$, that requires $\forall w\in U(1), U_{\vb*{\Theta} \vb*{\Phi} \lambda}U^\dag_{\vb*{\Theta} \vb*{\Phi} \lambda}=I$.
\end{proof}

\begin{theorem}
\label{thm:GQSP_f_cond}
Let $f(w)\coloneqq f(\theta = \arg w)$ be a Laurent polynomial with complex coefficients. Then, $\forall d\in \mathbb{N}, \exists \vb*{\Phi}\in \mathbb{R}^{d+1}$ s.t. $F = f$, \black{$U_{\vb*{\Theta} \vb*{\Phi} \lambda}(w) = V_{F G}(w)$ for $w \in U(1)$} if and only if
\begin{enumerate}[label=(\roman*)]
\item
$d_{{\rm max}, f} \leq d_+$ and $d_{{\rm min}, f} \geq -d_-$.
\item
$\forall w \in U(1), \abs{f(w)} \leq 1$.
\end{enumerate}
\end{theorem}
\begin{proof}
    The ``if" part is proved by explicitly constructing an angle finding algorithm in Sec.~\ref{sec:angle_finding_gqsp}. The ``only if" part follows in similar as in Theorem~\ref{thm:GQSP_FG_cond}.
\end{proof}

\section{Angle Finding for Generalized QSP} \label{sec:angle_finding_gqsp}
For the extended GQSP, we prove Theorems~\ref{thm:GQSP_FG_cond},~\ref{thm:GQSP_f_cond} by explicitly constructing an angle finding algorithm.
Here, we follow the overall structure of the ordinary QSP as in Fig.~\ref{fig:flow_chart} and explicitly provide procedures to perform \black{truncation, partition, and angle finding}.

\subsection{Truncation and Partition}\label{subsec:gqsp_truncation}
The truncation is performed so that a given complex function $f$ is approximated by a complex-coefficient Laurent polynomial $\tilde{f}$ with $d_{{\rm max}, \tilde{f}} \leq d_+$ and  $d_{{\rm min}, \tilde{f}} \geq -d_-.$
In similar to the ordinary QSP, some target functions allow analytical construction (e.g. the Hamiltonian simulation as is discussed in Sec.~\ref{sec:numerics}), while we may employ the numerical technique called the Remez exchange algorithm in general~\cite{remez1934, cheney2009course, Dong_2021}.
Subsequently, we perform the partition as
\begin{equation}
    \tilde{f}(w)=\alpha f_1(w),
\end{equation}
where
 $\alpha>0$ is taken sufficiently large so that $\forall w\in U(1), \abs{f_1(w)}<1.$
 Similarly to the case with the ordinary QSP, now the problem is reduced to computing the phase factors for a given function $f_1$, which satisfy the conditions given by Theorem~\ref{thm:GQSP_f_cond}.

\subsection{Computation of phase angles}
\black{Now we proceed to describe the computation of phase angles. 
For direct methods, as is the case for the ordinary QSP, the computation consists of the completion and decomposition.
Let us briefly describe the root finding and Prony's method for completion, and then discuss the decomposition.}

\subsubsection{Completion via root finding}
We describe how to determine the polynomials $F, G$ that meet the conditions of Theorem \ref{thm:GQSP_FG_cond}, assuming that a function $f$ satisfying the conditions in Theorem \ref{thm:GQSP_f_cond} is provided. In the following, we take
$F(w) \coloneqq f(w) \in \LaurentC{w}$ and assume that for all $w \in U(1), \abs{F(w)} = \abs{f(w)} < 1$.

First,  if $g(w) \coloneqq 1 - F(w) F^* \qty(w^{-1}) \in \LaurentC{w}$ has a set of $2 (d_- + d_+)$ roots denoted by $\qty{\xi_j}$, then with a constant $C_1 \in \C$, $g(w)$ can be expressed as
\begin{equation}
g(w) = C_1 w^{- (d_- + d_+)} \prod_{j = 1}^{2 (d_- + d_+)} \qty(w - \xi_j).
\end{equation}
Since $\forall w \in U(1), \abs{F(w)} \neq 1$, it follows that $\xi_j \notin U(1)$.
Note that $\abs{F(1 / w^*)} = \abs{F(w)}$ implies $g(1 / w^*) = g(w)$ for $w \in U(1)$. By the identity theorem, it follows that $g(1 / w^*) = g(w)$ also holds for $w \in \C \setminus \qty{0}$.
Therefore, with a constant $C_2 \in \C$, $g(w)$ can be written as
\begin{eqnarray}
g(w)
&=& C_1 w^{- (d_- + d_+)} \prod_{\abs{\xi_j} < 1} \qty(w - \xi_j) \qty(w - \frac{1}{\xi_j^*}) \nonumber \\
&=& C_2 \prod_{\abs{\xi_j} < 1} \qty(w - \xi_j) \prod_{\abs{\xi_j} < 1} \qty(w^{-1} - \xi_j^*) \nonumber \\
&=& C_2 \abs{\prod_{\abs{\xi_j} < 1} \qty(w - \xi_j)}^2.
\end{eqnarray}
Since $g(w) > 0$ for $w \in U(1)$, it is clear that $C_2 > 0$.

Thus, by defining $G(w) \coloneqq w^{- d_-} \sqrt{C_2} \prod_{\abs{\xi_j} < 1} \qty(w - \xi_j)$, $F$ and $G$ satisfy the conditions of Theorem \ref{thm:GQSP_FG_cond}.
In other words, by finding all the roots of $w^{(d_- + d_+)} \qty(1 - F(w) F^* \qty(w^{-1}))$ and comparing the value at $w = 1$ to determine $C_2$, one can compute the polynomials $F$ and  $G$.

\subsubsection{Completion via Prony's method}\label{subsubsec:prony_gqsp}

Similar to the Completion using Prony's method for $(W_z, S_x)$-QSP, $\prod_{\abs{\xi_j} < 1} \qty(w - \xi_j)$ is determined from $g(w) \coloneqq 1 - F(w) F^* \qty(w^{-1}) \in \LaurentC{w}$.

First, choose $N_s$ that is sufficiently larger than $d_- + d_+$, and 
define $t_n = \exp(i \frac{2 \pi n}{N_s})$. After calculating the Fourier coefficients $\qty{\hat{h}_k}$ of \black{$h(w)\coloneqq (g(w))^{-1}$} using $\qty{h(t_n)}$, find a vector $(m_0, \ldots, m_{d_- + d_+})$ that satisfies the following equation:
\begin{equation}
\mqty(
\hat{h}_{- 1} & \hat{h}_{- 2} & \cdots & \hat{h}_{- (d_- + d_+ + 1)} \\
\hat{h}_{- 2} & \hat{h}_{- 3} & \cdots & \hat{h}_{- (d_- + d_+ + 2)} \\
\vdots & \vdots & \ddots & \vdots \\
\hat{h}_{- l} & \hat{h}_{- (l + 1)} & \cdots & \hat{h}_{- (d_- + d_+ + l)}
)
\mqty(m_0 \\ \vdots \\ m_{d_- + d_+})
= 0,
\end{equation}
where $l$ is chosen such that $l \geq d_- + d_+ + 1$.
Then, one can see that
$G(w) \coloneqq \sqrt{C_2} \sum_{k = 0}^{d_- + d_+} m_k w^{k - d_-}$ ensures that $F$ and  $G$ satisfy the conditions of Theorem~\ref{thm:GQSP_FG_cond}.

In practice, we desire to
obtain $\qty(m_0, \ldots, m_{d_- + d_+})$ in a numerically stable way. For this purpose, instead of finding a QSP operation sequence for a function $f$ that satisfies Theorem \ref{thm:GQSP_f_cond}, we consider QSP operation sequences given by functions $f_1, f_2$ defined as:
\begin{align}
f_1(w) &\coloneqq \gamma \qty(w^{d_+} + w^{- d_-}), \\
f_2(w) &\coloneqq \frac{1}{\beta} f(w) - f_1(w),
\end{align}
where $\beta, \gamma > 0$ are chosen sufficiently large to ensure $\forall w \in U(1), \abs{f_i}<1~(i \in \{1,2\})$.
Since $f = \beta (f_1 + f_2)$, the function $f$ can be represented by a linear combination of QSP operation sequences for functions $f_1, f_2$.
By ensuring the highest and lowest degree coefficients of $f_2$ to be relatively large, we can prevent rank deficiency in the matrix used in Prony's method, which leads to the numerical stability.

\subsubsection{Decomposition}\label{subsec:decomposition_gqsp}

We note that we employ a simple trick to carry out the latter procedure  in parallel to that employed in Ref.~\cite{motlagh2023generalized}.
Namely, we multiply $w^{d_-}$ to the both sides of Eq.~\eqref{eq:gqsp_sequence} so that 
\begin{align}
w^{d_-} U_{\vb*{\Theta} \vb*{\Phi} \lambda}(w)
&= 
R(\theta_{- d_-}, \phi_{- d_-}, \lambda) 
\cdot \qty(\prod_{k = - d_- + 1}^0 W_0(w) R(\theta_k, \phi_k, 0)) 
\cdot \qty(\prod_{k = 1}^{d_+} W_0(w) R(\theta_k, \phi_k, 0)), \\
w^{d_-} V_{F G}(w)
&= \mqty(w^{d_-} F(w) & i w^{d_-} G(w) \\ \cdot & \cdot).
\end{align}
Note that the lowest power of $w$ in $w^{d_-} F$ and $w^{d_-} G$ are no less than zero. We find that the carving algorithm for decomposition can be straightforwardly applied when $d_- = 0$.
\black{For further details, we guide the readers to Appendix~\ref{app:angle_finding_gqsp}.}

\section{Numerical Demonstration}\label{sec:numerics}
Here, we perform a numerical benchmark to verify the numerical stability of the proposed algorithms.
In particular, we consider the task of Hamiltonian simulation, which can be understood as constructing QSP sequence for $f(x) = e^{-i \tau x}.$

\subsection{Hamiltonian simulation}\label{subsec:ham_sim}
The truncation (and partition) can be done analytically by considering the Jacobi-Anger expansion of $f(x)$ as
\begin{equation}
e^{- i \tau x}
=
J_0(\tau)
+ 2 \sum_{k \colon \textrm{even}} (- 1)^{k / 2} J_k(\tau) T_k(x) \\
+ 2 i \sum_{k \colon \textrm{odd}} (- 1)^{(k - 1) / 2} J_k(\tau) T_k(x),
\end{equation}
where $J_k$ is the Bessel function of the first kind and $T_k$ is the Chebyshev polynomial of the first kind.
By rewriting the expression using $w \coloneqq e^{i \arccos x} \in U(1)$, we obtain
\begin{equation}
e^{- i \tau x}
=
J_0(\tau)
+ \sum_{k \colon \textrm{even}} (- 1)^{k / 2} J_k(\tau) \qty(w^{k} + w^{- k}) \\
+ i \sum_{k \colon \textrm{odd}} (- 1)^{(k - 1) / 2} J_k(\tau) \qty(w^{k} + w^{- k}).
\end{equation}

\black{
In the truncation part, at even order $d$, equation (24) or equation (25) is truncated as follows:
\begin{align}
e^{- i \tau x}
&=
J_0(\tau)
+ 2 \sum_{k \colon \textrm{even}, \abs{k} \leq d} (- 1)^{k / 2} J_k(\tau) T_k(x)
+ 2 i \sum_{k \colon \textrm{odd}, \abs{k} \leq d - 1} (- 1)^{(k - 1) / 2} J_k(\tau) T_k(x), \\
e^{- i \tau x}
&=
J_0(\tau)
+ \sum_{k \colon \textrm{even}, \abs{k} \leq d} (- 1)^{k / 2} J_k(\tau) \qty(w^{k} + w^{- k})
+ i \sum_{k \colon \textrm{odd}, \abs{k} \leq d - 1} (- 1)^{(k - 1) / 2} J_k(\tau) \qty(w^{k} + w^{- k}).
\end{align}
In particular, in GQSP, we set $d = d^+ = d^-$.
}

We comment on the implementation cost in terms of the number of calls to the controlled signal operators.
By noting that the partition operation divides $e^{-i \tau x}$ into two functions,
\black{
\begin{align}
f_1(w)
&=
J_0(\tau)
+ \sum_{k \colon \textrm{even}, \abs{k} \leq d} (- 1)^{k / 2} J_k(\tau) \qty(w^{k} + w^{- k}), \\
f_2(w)
&=
i \sum_{k \colon \textrm{odd}, \abs{k} \leq d - 1} (- 1)^{(k - 1) / 2} J_k(\tau) \qty(w^{k} + w^{- k}),
\end{align}
}
most algorithms based on the ordinary QSP (namely $(W_x, S_z)$-QSP and $(W_x, S_z)$-QSP) require \black{$2d + (2d - 2) = 4d - 2$} calls to controlled-$U$ or controlled-$U^\dagger$ in order to realize the truncation order of $d$.
This is \black{$2(4d - 2) = 8d - 4$} in algorithms that employ Prony's method or gradient-based optimization, since one must extract the real part of $U_{\vb*{\Phi}}$ \black{by constructing the linear combination $(U_{\vb*{\Phi}} + U_{\vb*{\Phi}}^\dag) / 2$}.
On the other hand, these numbers are approximately halved when we employ GQSP; the root-finding-based algorithm requires $2d$ calls, whereas the direct method with Prony's method requires $4d$ calls \black{, depending on how many functions are divided in the partition part}.
\black{Note that this is due to the fact that we are able to constitute complex functions and thus the number of partitions are reduced accordingly.}
The optimization-based algorithm also requires $2d$ calls for GQSP.
\black{
Although the cost of amplitude amplification is not included here, it should be noted that the scaling factor is taken as $1/2$ for all methods in the following numberical demonstration, so the difference in cost and accuracy due to amplitude amplification among the methods can be considered negligible.
}

\subsection{Results}\label{subsec:results}

In order to demonstrate the accuracy of angle finding algorithms, we compute the implementation error of Hamiltonian simulation of $\tau=10,$ \black{$30$, and $100$}.
For the sake of benchmark, here we define the implementation error as 
\begin{eqnarray}
    \epsilon \coloneqq \|\tilde{f}_{\rm QSP} - f\|_{\infty},
\end{eqnarray}
where $\tilde{f}_{\rm QSP}$ is the approximation of $e^{- i \tau x}$ obtained from each angle finding algorithms.

Figure~\ref{fig:ham_sim_final_error}(a) shows the total error under various truncation order $d$ for $\tau=10$. As is consistent with previous works, the error at small $d$ is dominated by the truncation of the target function, while at larger $d$ the angle finding algorithms themselves are the dominant source of the error. For instance, we can see from Fig.~\ref{fig:ham_sim_final_error}(a) that the direct methods for the ordinary QSP that relies on root finding all suffers to achieve total error below $\epsilon=10^{-9}$ due to errors in completion and decomposition; one must rely on either the Prony's method or the gradient-based optimizations to achieve accuracy beyond $10^{-10}$, both of which requires $8d-4$ calls to controlled signal operator.

We find that our modified formalism of GQSP allows us to achieve the total error of $10^{-13}$ with approximately halved query complexity \black{(see Figs.~\ref{fig:ham_sim_final_error} and~\ref{fig:ham_sim_final_error2} for results of $\tau=10$ and $\tau=30, 100,$ respectively)}. 
This is done by employing the Prony's method tailored for GQSP (Sec.~\ref{subsubsec:prony_gqsp}), which requires $4d$ calls to controlled signal operator to implement the target function with truncation order of $d$. We observe that, at low or intermediate-precision regime, one may possibly benefit by employing the conventional root finding scheme, while it is an open problem how stable this approach is in general.
It is crucial to mention the poor performance of the optimization-based method for GQSP, as opposed to the stable performance for the ordinary QSP.
Here we find that the reflection symmetry of Pauli rotation in the ordinary QSP is crucial; Ref.~\cite{Dong_2021} explicitly utilized the symmetry in phase factors $\vb*{\Phi}$ to assure the convergence to the global optimum.
\black{
In the case where the symmetry exists, starting from a special symmetric initial guess
\begin{equation}
\Phi^0 = \qty(\frac{\pi}{4}, 0, 0, \ldots, 0, 0, \frac{\pi}{4}),
\end{equation}
using standard unconstrained optimization algorithms, one can find one global minima with positive definite Hessians in the neighborhood of $\Phi^0$ \cite{Wang2022energylandscapeof}.
However, such a symmetry is in general absent in GQSP so that similar technique cannot be adopted.
It can be inferred from the G.O. graph in Figure~\ref{fig:ham_sim_final_error}(b) that, with a random initial guess, the likelihood of finding a low-error solution is nearly zero.
}

Finally, we remark on the runtime of the algorithms.
We expect that the preprecessing time for the GQSP is affordable, since the classical runtime for the root finding algorithms are only increased by constant factors compared to those of QSP.
For instance, Fig.~\ref{fig:ham_sim_time} shows that the runtime scales polynomially (linear to quadratic) with $d$ for both the ordinary QSP and GQSP, and the difference in the prefactor is less than a factor of ten. We observe that these scaling are better than the formal one: direct methods in general runs with cubic scaling with $d$, while in practice we observe linear dependence for $d$ up to a few hundreds. It is possible that the bottleneck of cubic scaling dominates the runtime at larger values of $d$.
\black{We remark that these results are not inconsistent with Ref.~\cite{motlagh2023generalized} which discuss the runtime to perform the completion step for random polynomials, not the entire angle finding procedure for practical QSP tasks such as Hamiltonian simulation. We have also found that the accuracy of completion by the FFT-based convolution mentioned in Ref.~\cite{motlagh2023generalized} is poor so that we cannot employ for the task studied in this work.}

\begin{figure}[tb]
    \begin{center}
        \includegraphics[width=0.9\linewidth]{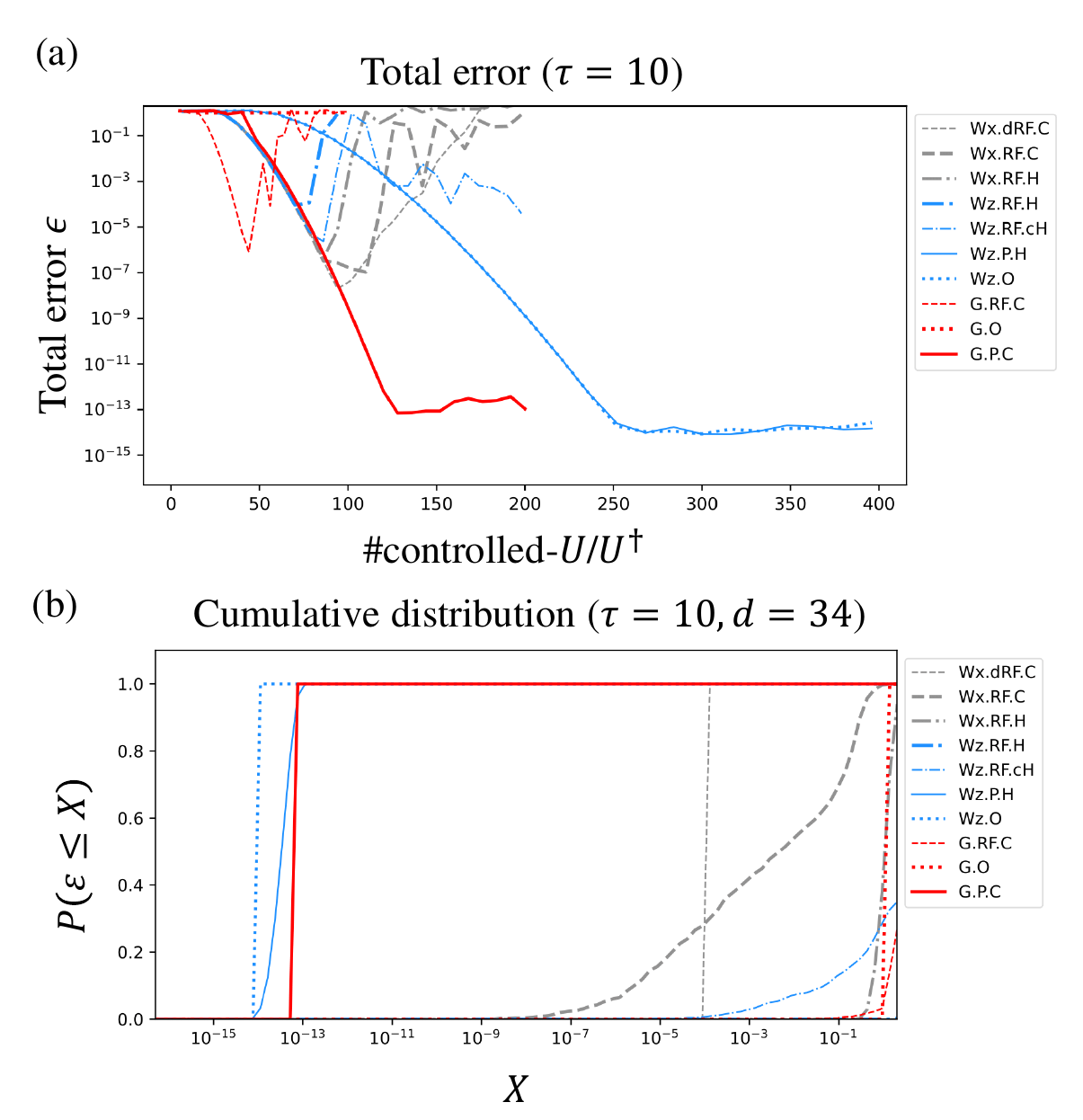}
        \caption{Errors in angle finding of QSP and GQSP for Hamiltonian simulation of $\tau=10$. Here, we compare the results between various choice of QSP bases with direct/optimization methods. (a) Scaling of the total error with respect to the number of controlled signal operators. We take  the best result among 10 trials, while the median values exhibit similar behavior. Here we take $d_+=d_-$ for GQSP. (b) Cumulative distribution of the total error $\epsilon$ for $d=d_+=d_-=34$ over 1000 trials.
        For both panels, {\bf Wx}, {\bf Wz}, {\bf G} denotes the $(W_x, S_z)$-QSP, $(W_z, S_x)$-QSP, and GQSP modified from Ref.~\cite{motlagh2023generalized} as Eq.~\eqref{eq:gqsp_sequence}. For the direct methods, we denote the completion step by {\bf RF/dRF} the root finding with random/deterministic protocol, and by {\bf P} the Prony's method. The decomposition steps are indicated as {\bf C}, {\bf H}, and {\bf cH} for the carving, halving, and halving with captalization using $\epsilon_{\rm cap}=10^{-8}$. The optimization-based result is denoted by {\bf O}.}
        \label{fig:ham_sim_final_error}
    \end{center}
\end{figure}

\begin{figure*}[tb]
    \begin{center}        \includegraphics[width=0.9\linewidth]{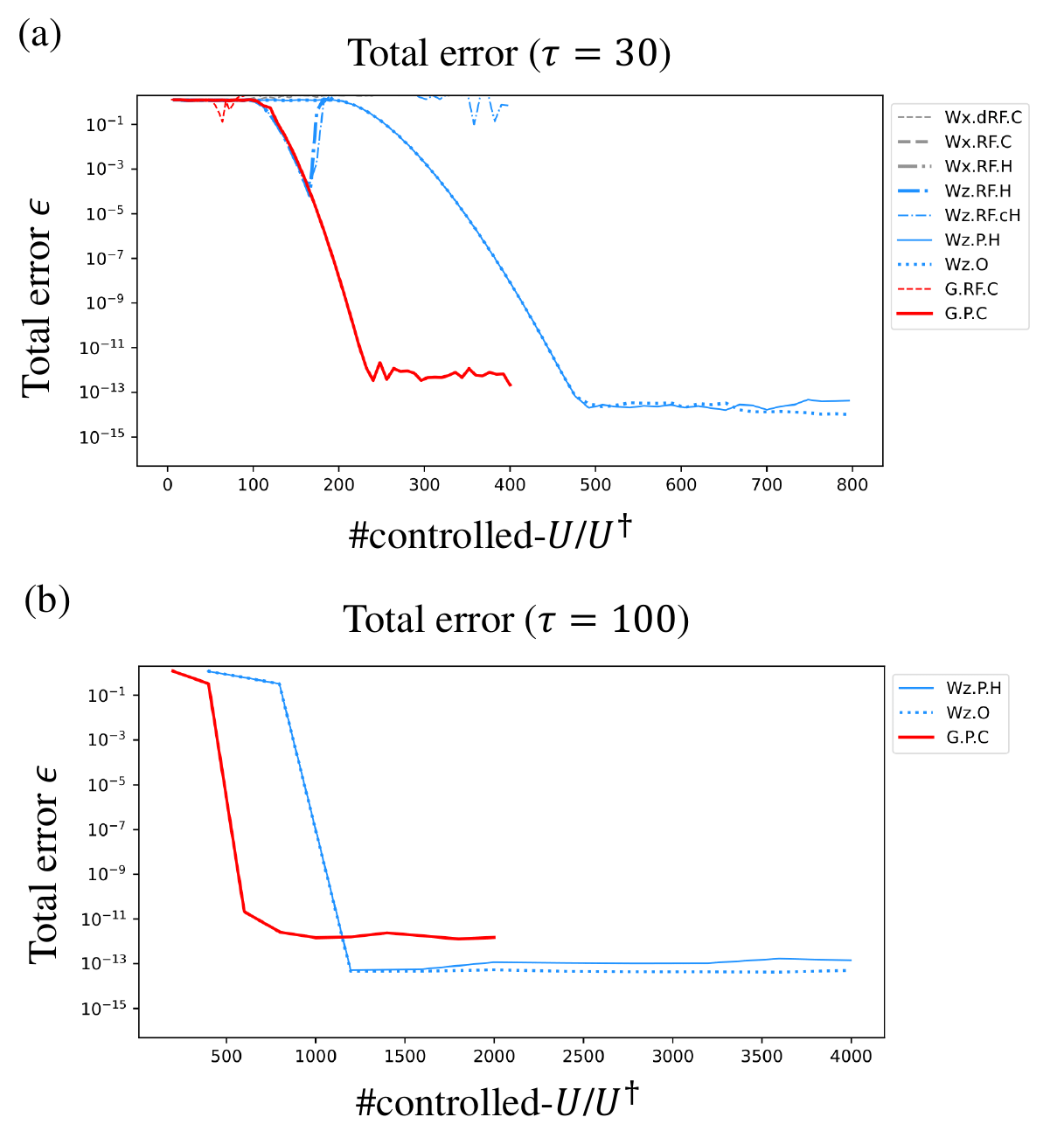}
        \caption{Errors in angle finding of QSP and GQSP for Hamiltonian simulation of (a) $\tau=30$, (b) $\tau=100$.
        For the methods not shown in the graph, either no correct solution was obtained or they failed to complete within the allotted time.
        }
        \label{fig:ham_sim_final_error2}
    \end{center}
\end{figure*}

\begin{figure*}[tb]
    \begin{center}        \includegraphics[width=0.9\linewidth]{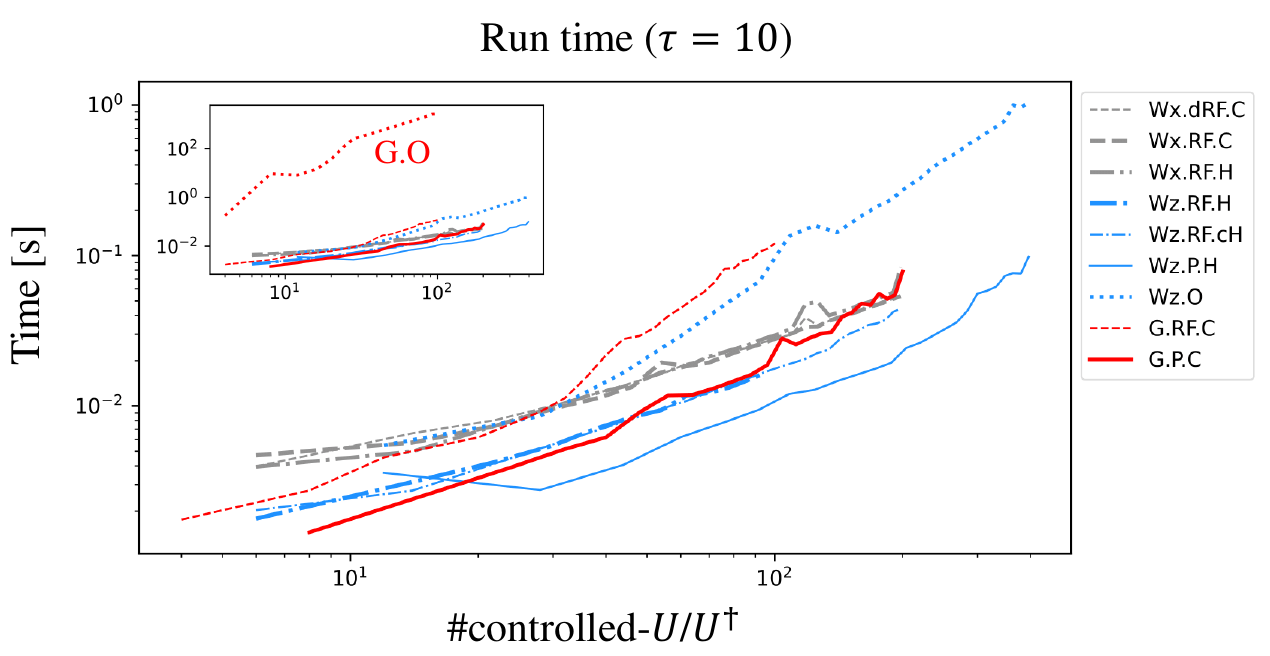}
        \caption{Runtime scaling of angle finding algorithms. Here, $d_+=d_-$ is assumed for GQSP. The notation of legends follow that of Fig.~\ref{fig:ham_sim_final_error}.
        The inset includes the data from optimization-based method for GQSP ({\bf G.O}), which is excluded in the main figure for better visibility. 
        }
        \label{fig:ham_sim_time}
    \end{center}
\end{figure*}

\section{Discussion}
In this work, we have extended the framework of generalized quantum signal processing (GQSP) proposed in Ref.~\cite{motlagh2023generalized}, and also proposed a robust algorithm that computes the phase angles up to high accuracy regime that is limited by machine precision.
By performing a numerical benchmark of computing phase sequence for  Hamiltonian simulation, we find that our scheme essentially halves the queries to signal operators compared to the ordinary QSP that can be readily applied at precision of $10^{-13}.$

As a future perspective, we envision that it is practically important to reveal how the algorithm performs under the presence of noise.
It has been pointed out that one may construct a systematic procedure to suppress the overrotation errors in the ordinary QSP~\cite{tan2023error, siegel2023algorithmic}, while it is nontrivial whether similar technique can be applied more efficiently for GQSP.
Also, it is interesting to seek for other advantages that can be obtained by utilizing the general form of QSP.

\section*{Acknowledgement}
The authors wishes to thank Danial Motlagh for sharing the code for angle finding employed in the first draft of Ref.~\cite{motlagh2023generalized}.
N.Y. wishes to thank
JST PRESTO No. JPMJPR2119, 
JST Grant Number JPMJPF2221, 
JST ERATO Grant Number  JPMJER2302, 
JST CREST Grant Number JPMJCR23I4,  
and the support from IBM Quantum.

\black{Numerical experiments are conducted using Python 3.11.1 and a laptop powered by AMD Ryzen 7 PRO 5850U CPU with 16 GB RAM.}
The codes are available via GitHub repository~\cite{github2024}.

\emph{Note added.---}
During the completion of our manuscript, we became aware
of an independent work by Berry {\it et al.}~\cite{berry2024doubling} that also pointed out that modification to the original form of generalized QSP in Ref.~\cite{motlagh2023generalized} allows one to reduce the cost of Hamiltonian simulation by a factor of 2.
\black{Also, we noticed that the original formalism of GQSP proposed by Motlagh and Wiebe in August 2023 has been extended by the revision done on late January 2024, which presents a formalism  equivalent to what we present in our work. The numerical demonstrations are focused on the runtime of completion procedure for random polynomials.}

\clearpage
\bibliography{bib.bib}

\clearpage

\appendix

\section{Details on angle finding algorithms for $(W_x, S_z)$-QSP} \label{app:angle_finding_Wx}


\subsection{Completion via root finding}

In this section, we \black{review} the procedure \black{in~\cite[Lemma 6]{Gily_n_2019}} for computing the polynomials $P, Q$ that satisfy the conditions of Theorem~\ref{thm:Wx_QSP_PQ_cond} given a function $f$ satisfying the conditions in Theorem~\ref{thm:Wx_QSP_f_cond}.
\black{Note that this completes the proof of ``if" part of Theorem~\ref{thm:Wx_QSP_f_cond}.}
For the subsequent discussion, let $f \in \R[x]$ be a polynomial such that $\deg(f) = l \leq d$, $\Parity(f) = d \bmod 2$ and $\abs{f(x)} \leq 1$ for all $x \in [-1, 1]$.

First, we introduce a real polynomial $A(x) \coloneqq 1 - f(x)^2 \in \R[x]$.
Since $f$ is a polynomial of degree $l$ with a definite parity, $A$ is an even polynomial of degree $2 l$, which can be easily verified as
\begin{eqnarray}
    \begin{cases}
A(- x)
= 1 - f(- x)^2
= 1 - f(x)^2
= A(x) & (f:\text{odd})\\
A(- x)
= 1 - f(- x)^2
= 1 - (- f(x))^2
= A(x) 
& (f:\text{even})
    \end{cases}.
\end{eqnarray}
By substituting $x = \sqrt{y}$ and defining $\tilde{A}(y) \coloneqq A(\sqrt{y})$, we observe that $\tilde{A}(y) \in \R[y]$ and $\deg(\tilde{A}) = l$.
\black{Let $S$ denote the multiset of roots of $\tilde{A}(y)$, considering the algebraic multiplicity.
\black{It can be shown that we can construct a partition of $S$ as}
\begin{eqnarray}
S = S_{(- \infty, 0)} \cup S_{[0, 1]} \cup S_{(1, \infty)} \cup S_H \cup \qty{s^* \colon s \in S_H},    
\end{eqnarray}
where we have defined the following subsets (submultisets) of $S$ defined as
\begin{align}
S_{(- \infty, 0)} &\coloneqq \qty{s \in S \colon s \in (- \infty, 0)}, \\
S_{[0, 1]} &\coloneqq \qty{s \in S \colon s \in [0, 1]}, \\
S_{(1, \infty)} &\coloneqq \qty{s \in S \colon s \in (1, \infty)}, \\
S_H &\coloneqq \qty{s \in S \colon \Im[s] > 0}.
\end{align}
Note that if $s \in S$ then also $s^* \in S$ since $A$ is a real polynomial.
}
Since $\tilde{A}(y) \geq 0$ for all $y \in [0, 1]$, each root in $S_{[0, 1]}$ appears an even number of times.
Therefore, a multiset $S'_{[0, 1]}$ exists where the multiplicity of each element is exactly half of that in $S_{[0, 1]}$.
Consequently, we can express $\tilde{A}$ using a constant $ C\in \R$ as 
\begin{equation}
\tilde{A}(y)
= C \prod_{s \in S_{(- \infty, 0)}} (y - s) \prod_{s \in S'_{[0, 1]}} (y - s)^2 \prod_{s \in S_{(1, \infty)}} (s - y) \prod_{s \in S_H} (y - s) (y - s^*).
\end{equation}

\black{In order to find a function $B$ such that $B(x)B^*(x) = \tilde{A}(x)$}, we rewrite the above terms using the functions $R_1(x; s)$ and $R_2(x; s)$ defined as
\begin{align}
R_1(x; s) &\coloneqq \sqrt{\abs{s - 1}} x + i \sqrt{\abs{s}} \sqrt{1 - x^2}, \\
R_2(x; s) &\coloneqq (\abs{s - 1} + \abs{s}) x^2 - \abs{s} + i \sqrt{(\abs{s - 1} + \abs{s})^2 - 1} \qty(x \sqrt{1 - x^2}).
\end{align}
For $s \in S_{(1, \infty)}$, $s \in S_{(- \infty, 0)}$ and $s \in S_H$, the following equations holds, respectively:
\begin{align}
R_1(x; s) R_1^*(x; s)
&= \qty(\sqrt{\abs{s - 1}} x + i \sqrt{\abs{s}} \sqrt{1 - x^2}) \qty(\sqrt{\abs{s - 1}} x - i \sqrt{\abs{s}} \sqrt{1 - x^2}) \nonumber \\
&= \abs{s - 1} x^2 + \abs{s} \qty(1 - x^2) \nonumber \\
&= (s - 1) x^2 + s \qty(1 - x^2) \nonumber \\
&= s - y, \\
R_1(x; s) R_1^*(x; s)
&= \abs{s - 1} x^2 + \abs{s} \qty(1 - x^2) \nonumber \\
&= - (s - 1) x^2 - s \qty(1 - x^2) \nonumber \\
&= y - s, \\
R_2(x; s) R_2^*(x; s)
&=
\begin{multlined}[t]
\qty((\abs{s - 1} + \abs{s}) x^2 - \abs{s} + i \sqrt{(\abs{s - 1} + \abs{s})^2 - 1} \qty(x \sqrt{1 - x^2})) \\
\cdot \qty((\abs{s - 1} + \abs{s}) x^2 - \abs{s} - i \sqrt{(\abs{s - 1} + \abs{s})^2 - 1} \qty(x \sqrt{1 - x^2}))
\end{multlined} \nonumber \\
&= \qty((\abs{s - 1} + \abs{s}) x^2 - \abs{s})^2 + \qty((\abs{s - 1} + \abs{s})^2 - 1) x^2 \qty(1 - x^2) \nonumber \\
&= x^4 + \qty(\abs{s - 1}^2 - \abs{s}^2 - 1) x^2 + \abs{s}^2 \nonumber \\
&= x^4 - \qty(s + s^*) x^2 + \abs{s}^2 \nonumber \\
&= (y - s) \qty(y - s^*).
\end{align}

Therefore $\tilde{A}$ can be rewrited as
\begin{equation}
\tilde{A}(y)
= C \prod_{s \in S_{(- \infty, 0)} \cup S_{(1, \infty)}} R_1(x; s) R_1^*(x; s) \prod_{s \in S'_{[0, 1]}} (y - s)^2 \prod_{s \in S_H} R_2(x; s) R_2^*(x; s).
\end{equation}
Since $\tilde{A}(y) \geq 0$ for all $y \in [0, 1]$, it is clear that the scaling factor $C$ is non-negative.
Thus, by defining a function $B(x)$ as
\begin{equation}
B(x)
= \sqrt{C} \prod_{s \in S_{(- \infty, 0)} \cup S_{(1, \infty)}} R_1(x; s) \prod_{s \in S'_{[0, 1]}} \qty(x^2 - s) \prod_{s \in S_H} R_2(x; s),
\end{equation}
we have $B(x) B^*(x) = \tilde{A}(y)$.

Here, $B(x)$ can be expressed as $B(x) = P'(x) + Q'(x) \sqrt{1 - x^2}$, where $P'$ is a real polynomial of degree $l$ with parity-($l \bmod 2$) and $Q'$ is a real polynomial of degree $l - 1$ with parity-($l - 1\bmod 2$).
Thus, if $l \equiv d \bmod 2$, by defining $P = f + i P'$ and $Q = i Q'$, $P$ and $Q$ satisfies the conditions of Theorem~\ref{thm:Wx_QSP_PQ_cond}.
In case where $l \not\equiv d \bmod 2$, by redefining $B(x)$ as $B(x) \cdot \qty(x + i \sqrt{1 - x^2})$, $B(x) B^*(x) = \tilde{A}(y)$ still holds, while the degrees of $P$ and $Q$ increase by one and their parities invert.


\subsection{Decomposition via carving}

In this section, we describe the carving algorithm following Ref.~\cite[Theorem 3]{Gily_n_2019}, which is a consecutive process for computing the phase factors $\vb*{\Phi}$ from the polynomials $P$ and $Q$ that satisfy the conditions in Theorem~\ref{thm:Wx_QSP_PQ_cond}.
\black{Note that the existence of concrete algorithm for decomposition completes the ``if" part of Theorem~\ref{thm:Wx_QSP_PQ_cond}.
}

Let us begin by supposing $\deg(P) = 0$.
Due to the conditions of Theorem~\ref{thm:Wx_QSP_PQ_cond}, the polynomials $P$ and $Q$ take the form $P \equiv e^{i \phi_0}$ and $Q \equiv 0$ with some $\phi_0 \in \R$.
It follows that $\vb*{\Phi} = (\phi_0, \frac{\pi}{2}, - \frac{\pi}{2}, \ldots, \frac{\pi}{2}, - \frac{\pi}{2}) \in \R^{d + 1}$ is a feasible solution from the fact that $d$ is even and the following identity:
\begin{equation}
W_x(\theta) S_z \qty(\frac{\pi}{2}) W_x(\theta) S_z \qty(- \frac{\pi}{2})
= I.
\end{equation}

\black{Now we prove by induction on $d$ that we can recursively determine $\vb*{\Phi}$.}
\black{For the purpose of proving the existence of $\vb*{\Phi}$ by induction on $d$, we assume that there exists an angle sequence when the degree of polynomial is less than $d-1$.}
Since $P(x)^2 + \qty(1 - x^2) Q(x)^2 = 1$ as a polynomial, assuming that $0 < \deg(P) = l \leq d$, we must have that $\deg(Q) = l - 1$ and 
\begin{equation}
\abs{p_l}^2 - \abs{q_{l - 1}}^2 = 0,
\end{equation}
where $p_l$ and $q_{l - 1}$ denote the leading coefficients of $P$ and $Q$, respectively.
Therefore, there exists some $\phi_d \in \R$ such that $e^{2 i \phi_d} = \frac{p_l}{q_{l - 1}}$.

\black{Next, we define lower-order} polynomials  $\tilde{P}(x)$ and $\tilde{Q}(x)$ as
\begin{align}
\tilde{P}(x) &\coloneqq e^{- i \phi_d} \qty(x P(x) + \frac{p_l}{q_{l - 1}} \qty(1 - x^2) Q(x)), \\
\tilde{Q}(x) &\coloneqq e^{- i \phi_d} \qty(- P(x) + \frac{p_l}{q_{l - 1}} x Q(x)).
\end{align}
It can be shown that these polynomials satisfy
\begin{equation}
\mqty(\tilde{P}(x) & i \tilde{Q}(x) \sqrt{1 - x^2} \\ i \tilde{Q}^*(x) \sqrt{1 - x^2} & \tilde{P}^*(x)) W_x(\theta) S_z(\phi_d)
= U_{P Q}(x),
\end{equation}
\black{which is a consequence from the following calculation:}
\begin{align}
U_{P Q}(x) S_z(\phi_d)^{-1} W_x(\theta)^{-1}
&= \mqty(P(x) & i Q(x) \sqrt{1 - x^2} \\ i Q^*(x) \sqrt{1 - x^2} & P^*(x)) \mqty(e^{- i \phi_d} & 0 \\ 0 & e^{i \phi_d}) \mqty(x & - i \sqrt{1 - x^2} \\ - i \sqrt{1 - x^2} & x) \nonumber \\
&= \mqty(e^{- i \phi_d} \qty(x P(x) + e^{2 i \phi_d} \qty(1 - x^2) Q(x)) & i e^{- i \phi_d} \qty(- P(x) + e^{2 i \phi_d} x Q(x)) \sqrt{1 - x^2} \\ i e^{i \phi_d} \qty(- P^*(x) + e^{- 2 i \phi_d} x Q^*(x)) \sqrt{1 - x^2} & e^{i \phi_d} \qty(x P^*(x) +  e^{- 2 i \phi_d} \qty(1 - x^2) Q^*(x))) \nonumber \\
&= \mqty(\tilde{P}(x) & i \tilde{Q}(x) \sqrt{1 - x^2} \\ i \tilde{Q}^*(x) \sqrt{1 - x^2} & \tilde{P}^*(x)).
\end{align}
Furthermore, we can easily show that $\deg(\tilde{P}) \leq l - 1 \leq d - 1$, $\deg(\tilde{Q}) \leq l - 2 \leq d - 2$, $\Parity(P) = d - 1 \bmod 2$ and $\Parity(Q) = d - 2 \bmod 2$.
Due to the unitarity of $U_{P Q}$, $W_x(\theta)$ and $S_z(\phi_d)$, $U_{\tilde{P} \tilde{Q}}$ is \black{also} a unitary matrix.
Hence, $\tilde{P}$ and $\tilde{Q}$ satisfy the conditions of Theorem~\ref{thm:Wx_QSP_PQ_cond} for degree of $d - 1$, which, \black{by the assumption of the induction,} implies the existence of $\vb*{\tilde{\Phi}} = (\tilde{\phi}_0, \ldots, \tilde{\phi}_{d - 1}) \in \R^d$ such that
\begin{equation}
S_z(\tilde{\phi}_0) \prod_{k = 1}^{d - 1} W_x(\theta) S_z(\tilde{\phi}_k)
= \mqty(\tilde{P}(x) & i \tilde{Q}(x) \sqrt{1 - x^2} \\ i \tilde{Q}^*(x) \sqrt{1 - x^2} & \tilde{P}^*(x)).
\end{equation}
Consequently, $\vb{\Phi} = (\tilde{\phi}_0, \ldots, \tilde{\phi}_{d - 1}, \phi_d) \in \R^{d + 1}$ is a feasible solution for the given polynomials $P$ and $Q$.


\subsection{Decomposition via halving}

In this section, we describe the halving algorithm, which is another method to determine the phase factors $\vb*{\Phi}$ from the polynomials $P$ and $Q$ that meet the conditions in Theorem~\ref{thm:Wx_QSP_PQ_cond}.
\black{Note that, in contrast to the carving algorithm that essentially proves the ``if" part of Theorem~\ref{thm:Wx_QSP_PQ_cond}, the halving algorithm explicitly relies on Theorem~\ref{thm:Wx_QSP_PQ_cond}.
}

\black{It is important to note that, in the case of $\deg(P) < d$, we can embed a feasible solution for the low-degree case 
$(\phi_0, \ldots, \phi_{\deg(P)})$ as} $\vb*{\Phi} = \qty(\phi_0, \ldots, \phi_{\deg(P)}, \frac{\pi}{2}, - \frac{\pi}{2}, \ldots, \frac{\pi}{2}, - \frac{\pi}{2}) \in \R^{d + 1}$, since the following equation holds:
\begin{equation}
W_x(\theta) S_z \qty(\frac{\pi}{2}) W_x(\theta) S_z \qty(- \frac{\pi}{2})
= I.
\end{equation}
\black{Therefore, we hereafter} consider the case where $d = \deg(P)$ and the complex polynomials $P$, $Q$ are denoted by $P(x) = \sum_{k = 0}^d p_k x^k$ and $Q(x) = \sum_{k = 0}^{d - 1} q_k x^k$.

Let us first consider the case where $d = 0$.
Due to the conditions in Theorem~\ref{thm:Wx_QSP_PQ_cond}, the polynomials $P$ and $Q$ take the form $P \equiv e^{i \phi_0}$ and $Q \equiv 0$ with some $\phi_0 \in \R$.
Hence, $\vb*{\Phi} = (\phi_0)$ is a feasible solution.
If $d = 1$,
a feasible solution can be obtained as $\vb*{\Phi} = (\phi_0, \phi_1)$ with $\phi_0$ and $\phi_1$ given by
\begin{align}
\phi_0 &\coloneqq \frac{\arg(p_1) + \arg(q_0)}{2}, \\
\phi_1 &\coloneqq \frac{\arg(p_1) - \arg(q_0)}{2}.
\end{align}
\black{To see the validity of the solution, we can explicitly compute as} 
\begin{align}
U_{\vb*{\Phi}}
&= \mqty(e^{i \phi_0} & 0 \\ 0 & e^{- i \phi_0}) \mqty(x & i \sqrt{1 - x^2} \\ i \sqrt{1 - x^2} & x) \mqty(e^{i \phi_1} & 0 \\ 0 & e^{- i \phi_1}) \nonumber \\
&= \mqty(e^{i (\phi_0 + \phi_1)} x & i e^{i (\phi_0 - \phi_1)} \sqrt{1 - x^2} \\ i e^{- i (\phi_0 - \phi_1)} \sqrt{1 - x^2} & e^{- i (\phi_0 + \phi_1)} x) \nonumber \\
&= \mqty(p_1 x & i q_0 \sqrt{1 - x^2} \\ i q_0^* \sqrt{1 - x^2} & p_1^* x) \nonumber \\
&= U_{P Q}.
\end{align}

\black{Next, we prove by induction that $\vb*{\Phi}$ for $d\geq 2$ can be obtained from recursive decomposition of polynomials into lower-degree polynomials.}
\black{Our goal is to halve the polynomials $P$ and $Q$ into components of polynomials of degree less than $l$ and $d - l$ for $l \coloneqq \floor{d / 2}$.
To be precise, we aim to identify polynomials $P_1, Q_1, P_2, Q_2 \in \C[x]$ that meet all the conditions below:
\begin{enumerate}
\item[(a)] $P_1$, $Q_1$ satisfy the conditions in Theorem~\ref{thm:Wx_QSP_PQ_cond} for $l$.
\item[(b)] $P_2$, $Q_2$ satisfy the conditions in Theorem~\ref{thm:Wx_QSP_PQ_cond} for $d - l$.
\item[(c)] $U_{P_1 Q_1} U_{P_2 Q_2} = U_{P Q}$.
\item[(d)] $P_1(1) = 1$.
\end{enumerate}
If feasible solutions for $P_1$, $Q_1$ and $P_2$, $Q_2$ are obtained as $\qty(\phi_0, \phi_1, \ldots, \phi_l)$, $\qty(\phi'_0, \phi'_1, \ldots, \phi'_{d - l})$, respectively,
we can combine them as $\vb*{\Phi} = \qty(\phi_0, \phi_1, \ldots, \phi_l + \phi'_0, \phi'_1, \ldots, \phi'_{d - l})$ to constitute a feasible solution for polynomials $P$ and $Q$.
Therefore, the problem boils down to constructing the lower-order polynomials in an appropriate manner.
}

Using a complex vector $(x_m, \ldots, x_{l - 2}, x_l, y_n, \ldots, y_{l - 3}, y_{l - 1}) \in \C^{2 l + 1}$, let us express polynomials $P_1(x)$ and $Q_1(x)$  as
\begin{align}
P_1(x) &= \sum_{\substack{\abs{k} \leq l, \\ k \equiv l \bmod 2}} x_k w^k, \\
Q_1(x) &= \sum_{\substack{\abs{k} \leq l - 1, \\ k \equiv l - 1 \bmod 2}} y_k w^k,
\end{align}
where $m$ and $n$ denote $l \bmod 2$ and $(l - 1) \bmod 2$, respectively.
Since we must have the equations $U_{P_1 Q_1} U_{P_2 Q_2} = U_{P Q}$ and $U_{P_1 Q_1}^\dagger U_{P_1 Q_1} = I$, we can express $P_2$ and $Q_2$ with $P$, $Q$, $P_1$ and $Q_1$ as
\begin{align}
P_2(x) &= P_1^*(x) P(x) + \qty(1 - x^2) Q_1(x) Q^*(x), \\
Q_2(x) &= P_1^*(x) Q(x) - Q_1(x) P^*(x).
\end{align}
To fulfill the conditions (a) and (b) on the degrees of the polynomials, $\deg(P_2) \leq d - l$ and $\deg(Q_2) \leq d - l - 1$, the following equations must hold:
\begin{align}
\sum_{k_1 + k_2 = k} \qty(p_{k_2} x_{k_1}^* + \qty(q_{k_2}^* - q_{k_2 - 2}^*) y_{k_1}) &= 0, & (k = d + l, d + l - 2, \ldots, d - l + 2) \\
\sum_{k_1 + k_2 = k} \qty(q_{k_2} x_{k_1}^* - p_{k_2}^* y_{k_1}) &= 0. & (k = d + l - 1, d + l - 3, \ldots, d - l + 1)
\end{align}
Moreover, to satisfy the condition (d), we must have $\sum_k x_k^* = 1$.

By summarizing the above discussion, we can write down the conditions as
\begin{equation}
\label{eq:halving_conditiong_wx}
\mqty(M_{1 1} & M_{1 2} \\ M_{2 1} & M_{2 2} \\ M_{3 1} & M_{3 2})
\mqty(x_l^* \\ x_{l - 2}^* \\ \vdots \\ x_m^* \\ y_{l - 1} \\ y_{l - 3} \\ \vdots \\ y_n)
= \mqty(0 \\ \vdots \\ 0 \\ 1),
\end{equation}
where $M_{1 1}, \ldots, M_{3 2}$ are defined as
\begin{align}
M_{1 1} &\coloneqq \mqty(
p_d & 0 & 0 & \cdots & 0 \\
p_{d - 2} & p_d & 0 & \cdots & 0 \\
\vdots & \vdots & \ddots & \ddots & \vdots \\
p_{d - l + m} & p_{d - l + m + 2} & \cdots & p_{d - 2} & p_d \\
\vdots & \vdots & \ddots & \vdots & \vdots \\
p_{d - 2 l + 2} & p_{d - 2 l + 4} & \cdots & p_{d - l - m} & p_{d - l - m + 2}
), \\
M_{1 2} &\coloneqq \mqty(
- q_{d - 1}^* & 0 & 0 & \cdots & 0 \\
q_{d - 1}^* - q_{d - 3}^* & - q_{d - 1}^* & 0 & \cdots & 0 \\
\vdots & \vdots & \ddots & \ddots & \vdots \\
q_{d - l + n + 2}^* - q_{d - l + n}^* & q_{d - l + n + 4}^* - q_{d - l + n + 2}^* & \cdots & q_{d - 1}^* - q_{d - 3}^* & - q_{d - 1}^* \\
\vdots & \vdots & \ddots & \vdots & \vdots \\
q_{d - 2 l + 3}^* - q_{d - 2 l + 1}^* & q_{d - 2 l + 5}^* - q_{d - 2 l + 3}^* & \cdots & q_{d - l - n}^* - q_{d - l - n - 2}^* & q_{d - l - n + 2}^* - q_{d - l - n}^*
), \\
M_{2 1} &\coloneqq \mqty(
q_{d - 1} & 0 & 0 & \cdots & 0 \\
q_{d - 3} & q_{d - 1} & 0 & \cdots & 0 \\
\vdots & \vdots & \ddots & \ddots & \vdots \\
q_{d - l + m - 1} & q_{d - l + m + 1} & \cdots & q_{d - 3} & q_{d - 1} \\
\vdots & \vdots & \ddots & \vdots & \vdots \\
q_{d - 2 l + 1} & q_{d - 2 l + 3} & \cdots & q_{d - l - m - 1} & q_{d - l - m + 1}
), \\
M_{2 2} &\coloneqq \mqty(
- p_d^* & 0 & 0 & \cdots & 0 \\
- p_{d - 2}^* & - p_d^* & 0 & \cdots & 0 \\
\vdots & \vdots & \ddots & \ddots & \vdots \\
- p_{d - l + n + 1}^* & - p_{d - l + n + 3}^* & \cdots & - p_{d - 2}^* & - p_d^* \\
\vdots & \vdots & \ddots & \vdots & \vdots \\
- p_{d - 2 l + 2}^* & - p_{d - 2 l + 4}^* & \cdots & - p_{d - l - n - 1}^* & - p_{d - l - n + 1}^*
), \\
M_{3 1} &\coloneqq \mqty(1 & \cdots & 1), \\
M_{3 2} &\coloneqq \mqty(0 & \cdots & 0).
\end{align}
We can show that the conditions (a)-(d) are equivalent to the conditions (a)-(e) for the halving algorithm of $(W_z, S_x)$-QSP.
As we also mention in Appendix.~\ref{app:angle_finding_Wz}, due to Theorem~\ref{thm:Wx_QSP_PQ_cond} and the uniqueness of the angle sequence, it is shown that angle sequence exists if and only if Eq.~\eqref{eq:halving_conditiong_wx} is satisfied.
This proof guarantees \black{the necessity of the conditions (a)-(d); the solution of the equation Eq.~\eqref{eq:halving_conditiong_wx} uniquely exists if conditions are satisfied, up to scale.}
For proof of sufficiency \black{of the conditions,} see Ref.~\cite{chao2020finding}.

\section{Details on angle finding algorithms for $(W_z, S_x)$-QSP} \label{app:angle_finding_Wz}

\subsection{Completion via root finding}
\label{sec:wz_completion_via_root_finding}

\black{In this section, we describe the root finding algorithm for the task of completion, i.e.,} how to compute the Laurent polynomials $F, G$ that meet the conditions provided in Theorem~\ref{thm:Wz_QSP_FG_cond}, assuming that a function $f$ satisfying the conditions in Theorem~\ref{thm:Wz_QSP_f_cond} is given.
\black{Here we provide a more rigorous version of \cite{chao2020finding} and \cite{Haah_2019} in the sense that this algorithm provides a complete proof of the ``if" part of Theorem~\ref{thm:Wz_QSP_f_cond}, as it does not rely on the assumption that $\abs{f(w)}$ is {\it strictly} less than $1$ for $w \in U(1)$.}
In the following, let $f \in \LaurentR{w}$ be a Laurent polynomial with real coefficients such that $\deg(f) = l \leq d$, $\Parity(f) = d \bmod 2$ and $\abs{f(w)} \leq 1$ for all $w \in U(1)$.

First, we define $F(w) \coloneqq f(w)$ and introduce a Laurent polynomial with real coefficients $A(w) \coloneqq 1 - F(w) F \qty(w^{-1}) \in \LaurentR{w}$.
Since $f$ is a Laurent polynomial of degree $l$ with a definite parity, $w^{2 l} A(w)$ becomes a real polynomial and $A$ is an even Laurent polynomial of degree $2 l$, which can be easily verified as
\begin{equation}
    \begin{cases}
A(- w)
= 1 - F(- w) F \qty(- w^{-1})
= 1 - F(w) F \qty(w^{-1})
= A(w)
& (f:\text{odd}) \\
A(- w)
= 1 - F(- w) F \qty(- w^{-1})
= 1 - (- F(w)) \qty(- F \qty(w^{-1}))
= A(w)
& (f:\text{even})
    \end{cases}.
\end{equation}

Note that the equation $A(1 / w) = A \qty(w^*) = A \qty(1 / w^*) = A(w)$ holds for $w \in U(1)$.
Applying the identity theorem, we deduce that $A(1 / w) = A(w^*) = A(1 / w^*) = A(w)$ also holds for $w \in \C \setminus \qty{0}$.

Let $S$ denote the multiset of roots for the polynomial $w^{2 l} A(w) \in \R[w]$.
We introduce the following subsets (submultisets) of $S$ defined as
\begin{align}
S_{(-1, 1)} &\coloneqq \qty{s \in S \colon s \in (-1, 1)}, \\
S_{U(1)} &\coloneqq \qty{s \in S \colon \abs{s} = 1}, \\
S_D &\coloneqq \qty{s \in S \colon \abs{s} < 1, \Im[s] > 0},
\end{align}
for which we can construct a partition of $S$ as follows:
\begin{equation*}
S = S_{(-1, 1)} \cup S_{U(1)} \cup S_D \cup \qty{1 / s \colon s \in S_{(-1, 1)} \cup S_{U(1)} \cup S_D} \cup \qty{s^* \colon s \in S_D} \cup \qty{1 / s^* \colon s \in S_D}.
\end{equation*}

Given that $A(w) \geq 0$ for all $w \in U(1)$, the multiplicity of each root in $S_{U(1)}$ is an even number.
Thus, we can construct a multiset $S'_{U(1)}$ where the multiplicity of each element is precisely half of that found in $S_{U(1)}$.
This allows us to rewrite $A(w)$ using constants $C_1, C_2 \in \C$ as
\begin{align}
A(w)
&= C_1 w^{- 2 l} \prod_{s \in S_{(-1, 1)}} (w - s) (w - 1 / s) \prod_{s \in S'_{U(1)}} (w - s)^2 \prod_{s \in S_D} (w - s) (w - 1 / s) (w - s^*) (w - 1 / s^*) \nonumber \\
&= C_2 \prod_{s \in S_{(-1, 1)}} (w - s) (1 / w - s) \prod_{s \in S'_{U(1)}} (w - s) (1 / w - s^*) \prod_{s \in S_D} (w - s) (1 / w - s^*) (w - s^*) (1 / w - s).
\end{align}
Obviously, this can be rewritten for $w \in U(1)$ as
\begin{equation}
A(w)
= C_2 \prod_{s \in S_{(-1, 1)}} \abs{w - s}^2 \prod_{s \in S'_{U(1)}} \abs{w - s}^2 \prod_{s \in S_D} \abs{w - s}^2 \abs{w - s^*}^2.
\end{equation}
\black{Given that $A(w) \geq 0$ for all $w \in U(1)$,} it implies that we must have $C_2 \geq 0$.
\black{Thus, we have a polynomial $G$ that satisfies the conditions  $G(w) \in \LaurentR{w}$, $G(w) G \qty(w^{-1}) = 1 - F(w) F \qty(w^{-1})$ for $w \in U(1)$ and $\deg(G) = l \leq d$,
by taking $G(w)$ as}
\begin{equation}
G(w)
\coloneqq \sqrt{C_2} w^{- l} \prod_{s \in S_{(-1, 1)}} (w - s) \prod_{s \in S'_{U(1)}} (w - s) \prod_{s \in S_D} (w - s) (w - s^*).
\end{equation}
\black{Note that} all the roots of $A(w)$ will appear in pairs as $\qty{s, - s}$ since $A(w)$ is an even function. This implies that all the roots of $G(w)$ will also come in pairs $\qty{s, - s}$,
and thus $G(w)$ has a parity of $l \bmod 2$.
Consequently, if $l \equiv d \bmod 2$, the Laurent polynomials $F$ and $G$ meet all the conditions in Theorem~\ref{thm:Wz_QSP_FG_cond}.
In case where $l \not\equiv d \bmod 2$, by redefining $w G(w)$ with $G(w)$, $G(w) G \qty(w^{-1}) = 1 - F(w) F \qty(w^{-1})$ still holds and the parity is reversed.

The polynomial $G$ can be rewritten more concisely as
\begin{equation}
G(w)
= \sqrt{C_2} w^{- l} \prod_{s \in S, \abs{s} < 1} (w - s) \prod_{s \in S'_{U(1)}} (w - s),
\end{equation}
since the subsets of $S$ can be merged as
\begin{equation}
S_{(-1, 1)} \cup S_D \cup \qty{s^* \colon S_D} = \qty{s \in S \colon \abs{s} < 1}.
\end{equation}
By assuming $\abs{f(w)} < 1$ for $w \in U(1)$, we must have $S_{U(1)} = \emptyset$ and $G(w)$ can be easily calculated by the following equation:
\begin{equation}
G(w) = \sqrt{C_2} w^{- l} \prod_{s \in S, \abs{s} < 1} (w - s).
\end{equation}
Therefore, by finding all the roots of $A(w)$ in the unit circle and comparing the value at $w = 1$ to determine $C_2$, one can compute the polynomials $F$ and $G$.


\subsection{Completion via Prony's method}
\label{sec:wz_qsp_completion_via_Prony's_method}

In this section, we review \black{the completion algorithm via Prony's method provided in Ref.~\cite{Ying_2022}, which determines} the Laurent polynomials $F, G$ that satisfy the conditions of Theorem~\ref{thm:Wz_QSP_FG_cond} from a function $f$ that meets the conditions stated in Theorem~\ref{thm:Wz_QSP_f_cond}.
In the following, let $f \in \LaurentR{w}$ be a Laurent polynomial with real coefficients such that $\deg(f) = l \leq d$, $\Parity(f) = d \bmod 2$ and $\abs{f(w)} < 1$ for all $w \in U(1)$.
Note that we assume that $\abs{f(w)}$ is strictly less than one for $w \in U(1)$.

Let $\qty{\xi_j}$ denote a set of roots of $g(w)) \coloneqq 1 - \abs{f(w)}^2$, and we aim to calculate $\prod_{\abs{\xi_j} < 1} \qty(w - \xi_j)$ as proved in Section \ref{sec:wz_completion_via_root_finding}.
The polynomial $\prod_{\abs{\xi_j} < 1} \qty(w - \xi_j)$ is denoted by $m_0 w^0 + \cdots m_{2 l} w^{2 l}$, and we will determine the coefficients $(m_0, \ldots, m_{2 l})$.
Since the poles of the reciprocal $h(w) \coloneqq g(w)^{-1}$ are the roots of $g(w)$ denoted by $\qty{\xi_j}$, then with constants $\qty{c_j}$, $h(w)$ can be expressed as
\begin{equation}
h(w) = \sum_{\xi_j} \frac{c_j}{w - \xi_j} + \text{constant}.
\end{equation}

\black{In order to investigate the holistic behaviors of the poles of $h(w)$ in the unit circle while avoiding explicit computation of them,} let us consider the integrals
\begin{equation}
\frac{1}{2 \pi i} \int_{\abs{w} = 1} \frac{h(w)}{w^k} \frac{\dd{w}}{w}
\end{equation}
for integer values of $k \leq - 1$.
For a fixed $k \leq - 1$, by using the residue theorem at $\qty{\xi_i}$,
\begin{equation}
\frac{1}{2 \pi i} \int_{\abs{w} = 1} \frac{h(w)}{w^k} \frac{\dd{w}}{w}
= \frac{1}{2 \pi i} \sum_{\xi_j} \int_{\abs{w} = 1} - \frac{c_j w^{- (k + 1)}}{\xi_j - w} \dd{w}
= - \sum_{\abs{\xi_j} < 1} c_j \xi_j^{- (k + 1)}.
\end{equation}
On the other hand, with the Fourier coefficients of $h(t) \equiv h \qty(e^{i t})$ denoted by $\qty{\hat{h}_k}$, these integrals can be written as
\begin{equation}
\frac{1}{2 \pi i} \int_{\abs{w} = 1} \frac{h(w)}{w^k} \frac{\dd{w}}{w}
= \frac{1}{2 \pi} \int_0^{2 \pi} h(t) e^{- i k t} \dd{t}
= \hat{h}_k.
\end{equation}
Therefore, the following equations holds true:
\begin{equation}
\label{eq:wz_prony's_h_hat_minus}
\hat{h}_- \coloneqq \mqty(\hat{h}_{- 1} \\ \hat{h}_{- 2} \\ \vdots)
= \frac{1}{2 \pi i} \int_{\abs{w} = 1} h(w) \mqty(w^0 \\ w^1 \\ \vdots) \dd{w}
= \mqty(- \sum_{\abs{\xi_j} < 1} c_j \xi_j^0 \\ - \sum_{\abs{\xi_j} < 1} c_j \xi_j^1 \\ \vdots)
= - \sum_{\abs{\xi_j} < 1} c_j \mqty(\xi_j^0 \\ \xi_j^1 \\ \vdots).
\end{equation}

\black{It is useful to} introduce a shift operator $\mathcal{S} \colon (a_1, a_2, a_3, \ldots) \mapsto (a_2, a_3, a_4, \ldots)$ that shifts the elements of a semi-infinite vector to the left.
For any $\xi_j$ with $\abs{\xi_j} < 1$,
\begin{equation}
\mathcal{S} \mqty(\xi_j^0 \\ \xi_j^1 \\ \vdots) = S \mqty(\xi_j^1 \\ \xi_j^2 \\ \vdots).
\end{equation}
\black{Since the operators $\qty{\mathcal{S} - \xi_{j'}}$ all commute, for any $\xi_j$ with $\abs{\xi_j} < 1$,}
\begin{align}
\prod_{\abs{\xi_{j'}} < 1} (\mathcal{S} - \xi_{j'}) \mqty(\xi_{j'}^0 \\ \xi_{j'}^1 \\ \vdots)
&= \qty(\prod_{\abs{\xi_{j'}} < 1, j' \neq j} (\mathcal{S} - \xi_{j'})) (\mathcal{S} - \xi_j) \mqty(\xi_j^0 \\ \xi_j^1 \\ \vdots) \nonumber \\
&= \qty(\prod_{\abs{\xi_{j'}} < 1, j' \neq j} (\mathcal{S} - \xi_{j'})) \qty(\mqty(\xi_j^1 \\ \xi_j^2 \\ \vdots) - \mqty(\xi_j^1 \\ \xi_j^2 \\ \vdots)) \nonumber \\
&= 0
\end{align}
\black{As shown in Eq.~\eqref{eq:wz_prony's_h_hat_minus}, since $\hat{h}_-$ is a linear combination of the semi-infinite vectors $\qty{(\xi_j^0, \xi_j^1, \ldots)}$ with weights $\qty{- c_j}$,}
\begin{equation}
\prod_{\abs{\xi_{j'}} < 1} (\mathcal{S} - \xi_{j'}) \hat{h}_-
= \prod_{\abs{\xi_{j'}} < 1} (\mathcal{S} - \xi_{j'}) \qty(- \sum_{\abs{\xi_j} < 1} c_j \mqty(\xi_j^0 \\ \xi_j^1 \\ \vdots))
= - \sum_{\abs{\xi_j} < 1} c_j \prod_{\abs{\xi_{j'}} < 1} (\mathcal{S} - \xi_{j'}) \mqty(\xi_j^0 \\ \xi_j^1 \\ \vdots)
= 0.
\end{equation}
\black{Thus, we must have $m_0(\mathcal{S}^0 \hat{h}_-) + \cdots + m_{2 d} (\mathcal{S}^{2 d} \hat{h}_-) = 0$, and we can express the conditions on $(m_0, \ldots, m_{2 d})$ as} 
\begin{equation}
\mqty(
\hat{h}_{- 1} & \hat{h}_{- 2} & \cdots & \hat{h}_{- (2 d + 1)} \\
\hat{h}_{- 2} & \hat{h}_{- 3} & \cdots & \hat{h}_{- (2 d + 2)} \\
\vdots & \vdots & \ddots & \vdots
)
\mqty(m_0 \\ \vdots \\ m_{2 d})
= 0.
\end{equation}
\black{This follows from the fact that} the term $\mathcal{S}^k \hat{h}_-$, the vector obtained by applying the operator $\mathcal{S}$ to $\hat{h}_-$ for $k$ iterations, can be written as
\begin{equation}
\mathcal{S}^k \hat{h}_-
= \mqty(\hat{h}_{- k - 1} \\ \hat{h}_{- k - 2} \\ \vdots).
\end{equation}
In other words, by computing a non-zero vector in the null-space of the matrix above, one can obtain a possible vector $(m_0, \ldots, m_{2 d})$.
\black{In the actual numerical implementation}, only the first $2 d + 2$ rows of the semi-infinite matrix are required to determine $(m_0, \ldots, m_{2 d})$.
Namely, it is sufficient to consider the following equation:
\begin{equation}
\mqty(
\hat{h}_{- 1} & \hat{h}_{- 2} & \cdots & \hat{h}_{- (2 d + 1)} \\
\hat{h}_{- 2} & \hat{h}_{- 3} & \cdots & \hat{h}_{- (2 d + 2)} \\
\vdots & \vdots & \ddots & \vdots \\
\hat{h}_{- (2 d + 2)} & \hat{h}_{- (2 d + 3)} & \cdots & \hat{h}_{- (4 d + 2)}
)
\mqty(m_0 \\ \vdots \\ m_{2 d})
= 0.
\end{equation}
The existence of such solution is guaranteed by the proof in Appendix \ref{sec:wz_completion_via_root_finding}.


\subsection{Decomposition via halving}

In this section, \black{we overview the halving algorithm provided in  Ref.~\cite{chao2020finding}} that determines the phase factors $\vb*{\Phi}$ from the Laurent polynomials $F$ and $G$ that meet the conditions in Theorem~\ref{thm:Wz_QSP_FG_cond}.

As in the case for $(W_x, S_z)$-QSP,  if $\deg(F) < d$ is satisfied, one can embed a feasible solution as$\vb*{\Phi} = (\phi_0, \ldots, \phi_{\deg(F)},\\ \frac{\pi}{2}, - \frac{\pi}{2}, \ldots, \frac{\pi}{2}, - \frac{\pi}{2}) \in \R^{d + 1}$, as the following equation holds true:
\begin{equation}
W_z(\theta) S_x \qty(\frac{\pi}{2}) W_z(\theta) S_x \qty(- \frac{\pi}{2})
= I.
\end{equation}
\black{Therefore, it suffices to assume 
 $d = \deg(F)$} and the Laurent polynomials with complex coefficients $F$, $G$ are denoted by $F(w) = \sum_{k = - d}^d f_k w^k$ and $G(w) = \sum_{k = - d}^d g_k w^k$.

\black{If $d=0$, due to the conditions in Theorem~\ref{thm:Wz_QSP_FG_cond}, the Laurent polynomials $F$ and $G$ take the form $F \equiv \cos \phi_0$ and $G \equiv \sin \phi_0$ with some $\phi_0 \in \R$.
Hence, $\vb*{\Phi} = (\phi_0)$ is a feasible solution.
If $d = 1$ in turn, we can construct a solution as
 $\vb*{\Phi}=(\phi_0, \phi_1)$, where $\phi_0$ and $\phi_1$ are defined using $\theta_{\mathrm{sum}} \coloneqq \arg(F(1) + i G(1))$ and $\theta_{\mathrm{diff}} \coloneqq \arg(F(i) - i G(i)) - \frac{\pi}{2}$ as follows:}
\begin{align}
\phi_0 &\coloneqq \frac{\theta_{\mathrm{sum}} + \theta_{\mathrm{diff}}}{2}, \\
\phi_1 &\coloneqq \frac{\theta_{\mathrm{sum}} - \theta_{\mathrm{diff}}}{2}.
\end{align}
\black{To validate the solutions, we explicitly evaluate the unitaries for $w=1, i$:}
\begin{align}
U_{\vb*{\Phi}}(1)
&= \mqty(\cos \phi_0 & i \sin \phi_0 \\ i \sin \phi_0 & \cos \phi_0) \mqty(1 & 0 \\ 0 & 1) \mqty(\cos \phi_1 & i \sin \phi_1 \\ i \sin \phi_1 & \cos \phi_1) \nonumber \\
&= \mqty(\cos(\phi_0 + \phi_1) & i \sin(\phi_0 + \phi_1) \\ i \sin(\phi_0 + \phi_1) & \cos(\phi_0 + \phi_1)) \nonumber \\
&= \mqty(\cos \theta_{\mathrm{sum}} & i \sin \theta_{\mathrm{sum}} \\ i \sin \theta_{\mathrm{sum}} & \cos \theta_{\mathrm{sum}}) \nonumber \\
&= \mqty(F(1) & i G(1) \\ i G(1) & F(1)) \nonumber \\
&= U_{F G}(1), \\
U_{\vb*{\Phi}}(i)
&= \mqty(\cos \phi_0 & i \sin \phi_0 \\ i \sin \phi_0 & \cos \phi_0) \mqty(i & 0 \\ 0 & - i) \mqty(\cos \phi_1 & i \sin \phi_1 \\ i \sin \phi_1 & \cos \phi_1) \nonumber \\
&= \mqty(i \cos(\phi_0 - \phi_1) & \sin(\phi_0 - \phi_1) \\ - \sin(\phi_0 - \phi_1) & - i \cos(\phi_0 - \phi_1)) \nonumber \\
&= \mqty(i \cos \theta_{\mathrm{diff}} & \sin \theta_{\mathrm{diff}} \\ - \sin \theta_{\mathrm{diff}} & - i \cos \theta_{\mathrm{diff}}) \nonumber \\
&= \mqty(i (- i F(i)) & i G(i) \\ - (- i G(- i)) & - i (i F(- i))) \nonumber \\
&= \mqty(F(i) & i G(i) \\ i G(- i) & F(- i)) \nonumber \\
&= U_{F G}(i).
\end{align}
Note that we have used the following identities from third to forth lines in both evaluations,
\begin{align}
\cos \theta_{\mathrm{sum}} &= \Re[F(1) + i G(1)] = f_{-1} + f_1 = F(1), \\
\sin \theta_{\mathrm{sum}} &= \Im[F(1) + i G(1)] = g_{-1} + g_1 = G(1), \\
\cos \theta_{\mathrm{diff}} &= \Im[F(i) - i G(i)] = - f_{-1} + f_1 = - i F(i) = i F(- i), \\
\sin \theta_{\mathrm{diff}} &= - \Re[F(i) - i G(i)] = g_{-1} - g_1 = i G(i) = - i G(- i).
\end{align}
Since the degrees of $F$ and $G$ are less than one, the equality $U_{\vb*{\Phi}} = U_{F G}$ holds as polynomials.
Hence, $U_{\vb*{\Phi}} = U_{F G}$ for all $w \in U(1)$ and $\vb*{\Phi}$ is a feasible solution.

\black{Next, we prove by induction that, $\vb*{\Phi}$ for $d\geq 2$ can be obtained from recursive decomposition of polynomials into lower-degree polynomials.}
\black{Our aim is} to halve the Laurent polynomials $F$ and $G$ into components of Laurent polynomials of degree less than $l$ and $d - l$ for $l \coloneqq \floor{d / 2}$.
\black{To be precise, we desire to} find polynomials $F_1, G_1, F_2, G_2 \in \LaurentR{w}$ that satisfy all the conditions below:
\begin{enumerate}
\item[(a)] $F_1$, $G_1$ satisfy the conditions in Theorem~\ref{thm:Wz_QSP_FG_cond} for $l$.
\item[(b)] $F_2$, $G_2$ satisfy the conditions in Theorem~\ref{thm:Wz_QSP_FG_cond} for $d - l$.
\item[(c)] $U_{F_1 G_1} U_{F_2 G_2} = U_{F G}$.
\item[(d)] $F_1(1) = 1.$
\item[(e)] $G_1(1) = 0.$
\end{enumerate}
\black{Once the phase factors
 for $F_1$, $G_1$ and $F_2$, $G_2$ are obtained  as $\qty(\phi_0, \phi_1, \ldots, \phi_l)$ and $\qty(\phi'_0, \phi'_1, \ldots, \phi'_{d - l})$, respectively,
we can combine them as $\vb*{\Phi} = \qty(\phi_0, \phi_1, \ldots, \phi_l + \phi'_0, \phi'_1, \ldots, \phi'_{d - l})$ to constitute a feasible solution for polynomials $F$ and $G$. Therefore, the problem boils down to constructing the lower-order polynomials $F_1, G_1, F_2, G_2$ in an appropriate manner.
}


Using a complex vector $(x_{- l}, x_{- l + 2}, \ldots, x_l, y_{- l}, y_{- l + 2}, \ldots, y_l) \in \R^{2 l + 2}$, the Laurent polynomials $F_1(x)$ and $G_1(x)$ are expressed as
\begin{align}
F_1(w) &= \sum_{\substack{\abs{k} \leq l, \\ k \equiv l \bmod 2}} x_k w^k, \\
G_1(w) &= \sum_{\substack{\abs{k} \leq l, \\ k \equiv l \bmod 2}} y_k w^k.
\end{align}
\black{Due to the condition (c)} and the unitarity constraint $U_{F_1 G_1}^\dagger U_{F_1 G_1} = I$, we can express $F_2$ and $G_2$ with $F$, $G$, $F_1$ and $G_1$ as
\begin{align}
F_2(w) &= F_1 \qty(w^{-1}) F(w) + G_1(w) G \qty(w^{-1}), \\
G_2(w) &= F_1 \qty(w^{-1}) G(w) - G_1(w) F \qty(w^{-1}).
\end{align}
Therefore, to fulfill the \black{conditions (a) and (b)} on the degrees of the polynomials, the \black{following equations must hold}:
\begin{align}
\sum_{k_1 + k_2 = k} (f_{k_2} x_{- k_1} + g_{- k_2} y_{k_1}) &= 0, & (k = - d - l, - d - l + 2, \ldots, - d + l - 2) \\
\sum_{k_1 + k_2 = k} (f_{k_2} x_{- k_1} + g_{- k_2} y_{k_1}) &= 0, & (k = d - l + 2, d - l + 4, \ldots, d + l) \\
\sum_{k_1 + k_2 = k} (g_{k_2} x_{- k_1} - f_{- k_2} y_{k_1}) &= 0, & (k = - d - l, - d - l + 2, \ldots, - d + l - 2) \\
\sum_{k_1 + k_2 = k} (g_{k_2} x_{- k_1} - f_{- k_2} y_{k_1}) &= 0. & (k = d - l + 2, d - l + 4, \ldots, d + l)
\end{align}
Moreover, to satisfy the conditions \black{(d) and (e)},  we must have $\sum_k x_k = 1$ and $\sum_k y_k = 0$.

\black{By summarizing the above discussions, we can write down the conditions as}
\begin{equation}\label{eq:halving_conditiong_wz}
\mqty(M_{1 1} & M_{1 2} \\ M_{2 1} & M_{2 2} \\ M_{3 1} & M_{3 2} \\ M_{4 1} & M_{4 2} \\ M_{5 1} & M_{5 2})
\mqty(x_l \\ x_{l - 2} \\ \vdots \\ x_{- l} \\ y_{- l} \\ y_{- l + 2} \\ \vdots \\ y_l)
= \mqty(0 \\ \vdots \\ 0 \\ 1 \\ 0),
\end{equation}
where $M_{1 1}, \ldots, M_{5 2}$ are defined as
\begin{align}
M_{1 1} &\coloneqq \mqty(
f_{- d} & 0 & 0 & \cdots & 0 \\
f_{- d + 2} & f_{- d} & 0 & \cdots & 0 \\
\vdots & \vdots & \ddots & \ddots & \vdots \\
f_{- d + 2 l - 2} & f_{- d + 2 l - 4} & \cdots & f_{- d} & 0
), &
M_{1 2} &\coloneqq \mqty(
g_d & 0 & 0 & \cdots & 0 \\
g_{d - 2} & g_d & 0 & \cdots & 0 \\
\vdots & \vdots & \ddots & \ddots & \vdots \\
g_{d - 2 l + 2} & g_{d - 2 l + 4} & \cdots & g_d & 0
), \\
M_{2 1} &\coloneqq \mqty(
0 & f_d & f_{d - 2} & \cdots & f_{d - 2 l + 2} \\
0 & 0 & f_d & \cdots & f_{d - 2 l + 4} \\
\vdots & \vdots & \ddots & \ddots & \vdots \\
0 & 0 & \cdots & 0 & f_d
), &
M_{2 2} &\coloneqq \mqty(
0 & g_{- d} & g_{- d + 2} & \cdots & g_{- d + 2 l - 2} \\
0 & 0 & g_{- d} & \cdots & g_{- d + 2 l - 4} \\
\vdots & \vdots & \ddots & \ddots & \vdots \\
0 & 0 & \cdots & 0 & g_{- d}
), \\
M_{3 1} &\coloneqq \mqty(
g_{- d} & 0 & 0 & \cdots & 0 \\
g_{- d + 2} & g_{- d} & 0 & \cdots & 0 \\
\vdots & \vdots & \ddots & \ddots & \vdots \\
g_{- d + 2 l - 2} & g_{- d + 2 l - 4} & \cdots & g_{- d} & 0
), &
M_{3 2} &\coloneqq - \mqty(
f_d & 0 & 0 & \cdots & 0 \\
f_{d - 2} & f_d & 0 & \cdots & 0 \\
\vdots & \vdots & \ddots & \ddots & \vdots \\
f_{d - 2 l + 2} & f_{d - 2 l + 4} & \cdots & f_d & 0
), \\
M_{4 1} &\coloneqq \mqty(
0 & g_d & g_{d - 2} & \cdots & g_{d - 2 l + 2} \\
0 & 0 & g_d & \cdots & g_{d - 2 l + 4} \\
\vdots & \vdots & \ddots & \ddots & \vdots \\
0 & 0 & \cdots & 0 & g_d
), &
M_{4 2} &\coloneqq - \mqty(
0 & f_{- d} & f_{- d + 2} & \cdots & f_{- d + 2 l - 2} \\
0 & 0 & f_{- d} & \cdots & f_{- d + 2 l - 4} \\
\vdots & \vdots & \ddots & \ddots & \vdots \\
0 & 0 & \cdots & 0 & f_{- d}
), \\
M_{5 1} &\coloneqq \mqty(1 & \cdots & 1 \\ 0 & \cdots & 0), &
M_{5 2} &\coloneqq \mqty(0 & \cdots & 0 \\ 1 & \cdots & 1).
\end{align}
\black{Using the solutions of Eq.~\eqref{eq:halving_conditiong_wz}, we can construct lower-order polynomials $F_1, G_1, F_2, G_2$. By repeating this procedure until the degrees are 0 or 1, we can determine the phase sequence.
}
\black{Due to Theorem~\ref{thm:Wz_QSP_FG_cond} and the uniqueness of the angle sequence, it is shown that angle sequence exists if and only if Eq.~\eqref{eq:halving_conditiong_wz} is satisfied.
Thus, it is guaranteed that the solution of Eq.~\eqref{eq:halving_conditiong_wz} uniquely exists, up to scale~\cite{chao2020finding}.
}

\section{Details on angle finding algorithm for Generalized QSP}\label{app:angle_finding_gqsp}
\subsection{Completion via Prony's method}
In the proof of completion algorithm via Prony's method for $(W_z, S_x)$-QSP as discussed in Appendix~\ref{sec:wz_qsp_completion_via_Prony's_method}, the parity condition is not necessary and the coefficients of polynomials can be easily extended to the complex number domain.
Hence, we can prove that Prony's method can be applied to Generalized QSP in almost the same way as $(W_z, S_x)$-QSP.

\subsection{Decomposition via carving}

In this section, we describe \black{the carving algorithm introduced in Ref.~\cite[Theorem 3]{motlagh2023generalized}} that finds the phase factors $\vb*{\Phi}$ from the Laurent polynomials $F$ and $G$ satisfying the conditions in Theorem~\ref{thm:GQSP_FG_cond}.
\black{Note that the existence of concrete algorithm for the decomposition step completes the ``if" part of  Theorem~\ref{thm:GQSP_FG_cond}.}
\black{As mentioned in the main text, } the case where $d_- > 0$ can be reduced to the case where $d_- = 0$, so we assume $d_- = 0$ in the following discussion.

First, let us consider the trivial case of $d_{{\rm max}, F} = 0$.
Due to the conditions in Theorem~\ref{thm:GQSP_FG_cond}, the Laurent polynomials $F$ and $G$ can be expressed as $F \equiv e^{i (\lambda + \phi_0)} \cos \theta_0$ and $G \equiv - i e^{i \lambda} \sin \theta_0$ with some $\phi_0, \theta_0, \lambda \in \R$.
Therefore, $\vb*{\Theta} = \qty(- \frac{1}{2} \pi, \pi, \ldots, \pi, \theta_0 - \frac{1}{2} \pi), \vb*{\Phi} = (0, \pi, \ldots, \pi, \phi_0), \lambda$ is a feasible solution, as the following equation holds true:
\begin{align}
R \qty(- \frac{1}{2} \pi, 0, \lambda) & \qty(\prod_{k = 1}^{d_+ - 1} W_0(w) R(\pi, \pi, 0)) W_0(w) R \qty(\theta_0 - \frac{1}{2} \pi, \phi_0, 0) \nonumber \\
&= \mqty(0 & - e^{i \lambda} \\ 1 & 0) \qty(\prod_{k = 1}^{d_+ - 1} \mqty(w & 0 \\ 0 & 1) I) \mqty(w & 0 \\ 0 & 1) \mqty(e^{i \phi_0} \sin \theta_0 & - \cos \theta_0 \\ - e^{i \phi_0} \cos \theta_0 & - \sin \theta_0) \nonumber \\
&= \mqty(0 & - e^{i \lambda} \\ 1 & 0) \mqty(w^{d_+} & 0 \\ 0 & 1) \mqty(e^{i \phi_0} \sin \theta_0 & - \cos \theta_0 \\ - e^{i \phi_0} \cos \theta_0 & - \sin \theta_0) \nonumber \\
&= \mqty(e^{i (\lambda + \phi_0)} \cos \theta_0 & e^{i \lambda} \sin \theta_0 \\ w^{d_+} e^{i \phi_0} \sin \theta_0 & - w^{d_+} \cos \theta_0).
\end{align}
Note that, for $d_+ = 0$, we have $\vb*{\Theta} = (\theta_0)$, $\vb*{\Phi} = (\phi_0)$, $\lambda$ as a feasible solution.

\black{Now we prove by induction on $d_+$ that we can recursively determine $\vb*{\Phi}$.}
Let $f_{d_{{\rm max}/{\rm min}, F}}$ and $g_{d_{{\rm max}/{\rm min}, F}}$ denote the highest/lowest coefficients of $F$ and $G$.
Due to the condition (ii) in Theorem~\ref{thm:GQSP_FG_cond} \black{(i.e., $\abs{F(w)}^2 + \abs{G(w)}^2 = 1$ for $w \in U(1))$}, considering the highest coefficients of both sides, we must have $d_{{\rm max}, F} - d_{{\rm min}, F} = d_{{\rm max}, G} - d_{{\rm min}, G}$ and $f_{d_{{\rm max}, F}} f_{d_{{\rm min}, F}}^* + g_{d_{{\rm max}, G}} g_{d_{{\rm min}, G}}^* = 0$.
Thus, there exists some $\theta_{d_+}, \phi_{d_+} \in \R$ such that $(i g_{d_{{\rm max}, F}}) / f_{d_{{\rm max}, F}} = f_{d_{{\rm min}, F}}^* / (i g_{d_{{\rm min}, G}}^*) = e^{- i \phi_{d_+}} \tan \theta_{d_+}$.
\black{Next, we define lower-order polynomials} $\tilde{F}(w)$ and $\tilde{G}(w)$ as
\begin{align}
\tilde{F}(w) &\coloneqq (e^{- i \phi_{d_+}} \cos \theta_{d_+}) w^{-1} F(w) + (i \sin \theta_{d_+}) w^{-1} G(w), \\
\tilde{G}(w) &\coloneqq (- i e^{- i \phi_{d_+}} \sin \theta_{d_+}) F(w) + (- \cos \theta_{d_+}) G(w).
\end{align}
\black{It can be shown that these polynomials satisfy}
\begin{equation}
\mqty(\tilde{F}(w) & i \tilde{G}(w) \\ \cdot & \cdot) W_0(w) R(\theta_{d_+}, \phi_{d_+}, 0)
= V_{F G}(w),
\end{equation}
which is a consequence from the following calculation:
\begin{align}
& V_{F G}(w) R(\theta_{d_+}, \phi_{d_+}, 0)^{-1} W_0(w)^{-1} \\
&= \mqty(F(w) & i G(w) \\ \cdot & \cdot) \mqty(e^{- i \phi_{d_+}} \cos \theta_{d_+} & e^{- i \phi_{d_+}} \sin \theta_{d_+} \\ \sin \theta_{d_+} & - \cos \theta_{d_+}) \mqty(w^{-1} & 0 \\ 0 & 1) \nonumber \\
&= \mqty((e^{- i \phi_{d_+}} \cos \theta_{d_+}) w^{-1} F(w) + (i \sin \theta_{d_+}) w^{-1} G(w) & i \qty((- i e^{- i \phi_{d_+}} \sin \theta_{d_+}) F(w) + (- \cos \theta_{d_+}) G(w)) \\ \cdot & \cdot) \nonumber \\
&= \mqty(\tilde{F}(w) & i \tilde{G}(w) \\ \cdot & \cdot) \nonumber \\
&= V_{\tilde{F} \tilde{G}}.
\end{align}
We can easily show that $d_{{\rm max}, \tilde{F}} \leq d_{{\rm max}, F} - 1 \leq d_+ - 1$, $d_{{\rm max}, \tilde{G}} \leq d_{{\rm max}, G} - 1 \leq d_+ - 1$, $d_{{\rm min}, \tilde{F}} \geq 0$ and $d_{{\rm min}, \tilde{G}} \geq 0$.
Due to the unitarity of $U_{F G}$, $W_0(w)$ and $R(\theta_{d_+}, \phi_{d_+}, 0)$, $V_{\tilde{F} \tilde{G}}$ is also a unitary matrix.
Hence, $\tilde{F}$ and $\tilde{G}$ satisfy the conditions of Theorem~\ref{thm:GQSP_FG_cond} for $d_+ - 1$. By assumption of the induction, there exist  $\vb*{\Theta} = (\tilde{\theta}_0, \ldots, \tilde{\theta}_{d_+ - 1}) \in \R^{d_+}$, $\vb*{\Phi} = (\tilde{\phi}_0, \ldots, \tilde{\phi}_{d_+ - 1}) \in \R^{d_+}$, and $\tilde{\lambda} \in \R$ such that
\begin{equation}
R(\theta_0, \phi_0, \lambda) \prod_{k = 1}^{d_+ - 1} W_0(w) R(\theta_k, \phi_k, 0)
= \mqty(\tilde{F}(w) & i \tilde{G}(w) \\ \cdot & \cdot).
\end{equation}
Therefore,
 $\vb*{\Theta} = (\tilde{\theta}_0, \ldots, \tilde{\theta}_{d_+ - 1}, \theta_{d_+}) \in \R^{d_+ + 1}, \vb*{\Phi} = (\tilde{\phi}_0, \ldots, \tilde{\phi}_{d_+ - 1}, \phi_{d_+}) \in \R^{d_+ + 1}, \lambda = \tilde{\lambda} \in \R$ is a feasible solution for the given Laurent polynomials $F$ and $G$.


\end{document}